\documentclass{siamart0516}

\usepackage{amsfonts, amsmath, amssymb}
\usepackage{mathrsfs}
\usepackage{algorithm}
\usepackage{algorithmicx}
\usepackage{algpseudocode} 
\usepackage{afterpage}
\usepackage{graphicx}
\usepackage{datetime}

\newsiamthm{Theorem}{Theorem}
\newsiamthm{Definition}{Definition}
\newsiamthm{Lemma}{Lemma}
\newsiamthm{Corollary}{Corollary}
\newsiamthm{Observation}{Observation}

\ifpdf
  \DeclareGraphicsExtensions{.eps,.pdf,.png,.jpg}
\else
  \DeclareGraphicsExtensions{.eps}
\fi

\newcommand{\TheTitle}{A  $2/3$-Approximation Algorithm for Vertex Weighted Matching in Bipartite Graphs} 
\newcommand{\TheAuthors}{F. Dobrian, M. Halappanavar, A. Pothen, A. Al-Herz}

\headers{$2/3$-Approximation for Vertex Weighted Matching}{Dobrian, Halappanavar, Pothen and Al-Herz}

\title{{\TheTitle}\thanks{
\funding{We acknowledge support from 
NSF grants  CCF-1637534 and CCF-1552323; the 
U.S. Department of Energy through grant DE-FG02-13ER26135;  the Exascale Computing Project (17-SC-20-SC), a collaborative effort of the DOE Office of Science and the NNSA;  
and the Pacific Northwest National Laboratory,  operated by Battelle for the  
DOE under Contract DE-AC05-76RL01830. 
}}}

\author{
  Florin Dobrian\thanks{Conviva Corporation, 2 Waters Park, Suite 150, San Mateo CA 94403 (\email{dobrian@cs.odu.edu}).}
  \and
  Mahantesh Halappanavar\thanks{Pacific Northwest National Lab,  902 Battelle Blvd., PO Box 999, MSIN-J4-30, Richland WA 99352 (\email{mhala@pnnl.gov}).}
  \and
  Alex Pothen\thanks{Purdue University, Department of Computer Science, West Lafayette IN 47907 (\email{apothen@purdue.edu},
    \email{aalherz@purdue.edu}).}
  \and 
  \phantom{Ahmed} Ahmed Al-Herz \footnotemark[4]
}

\usepackage{amsopn}

\ifpdf
\hypersetup{
  pdftitle={\TheTitle},
  pdfauthor={\TheAuthors}
}
\fi

\newcommand{\chg}[1]{{\color{black}{#1}}}

\begin{document}

\maketitle

\begin{abstract}
  We consider the maximum vertex-weighted matching problem (MVM), 
in which non-negative weights are assigned to the vertices of a graph, 
the weight of a matching is the sum of the weights of the matched vertices, and we are required to compute a matching of maximum weight. 
We describe an exact algorithm for MVM with  $O(|V|\, |E|)$ time complexity, and then we 
design a $2/3$-approximation algorithm for 
MVM on bipartite graphs by restricting the length of augmenting paths to at most three. The latter algorithm  has time complexity 
$O(|E| + |V| \log |V|)$. 

The approximation algorithm solves  two MVM problems on bipartite graphs, 
each with weights only on one vertex part,  
and then finds a matching from  these two matchings 
using the Mendelsohn-Dulmage Theorem. 
The approximation ratio of the algorithm is obtained by considering  failed vertices, 
i.e., vertices that the approximation algorithm
fails to match but the exact algorithm does. 
We show that at every step of the algorithm
there are two distinct heavier vertices that we can charge each  failed vertex to.

We have implemented the $2/3$-approximation algorithm for MVM and compare it with four other algorithms: 
an exact MEM algorithm, the exact MVM algorithm, a $1/2$-approximation algorithm for MVM, and a scaling-based  $(1-\epsilon)$-approximation algorithm for MEM. On a test set of nineteen problems with several millions of vertices, we show that the  maximum time taken by the exact MEM algorithm is $15$ hours, while it is $22$ minutes for the exact MVM algorithm, and less than $5$ seconds for the $2/3$-approximation algorithm. 
The $2/3$-approximation algorithm  obtains more than $99.5\%$ of the weight and cardinality  of an MVM, whereas the scaling-based approximation algorithms yield lower weights and cardinalities while taking an order of magnitude more time than the former algorithm. 

We also show that  MVM problems should not  be first transformed to MEM problems and solved using exact algorithms for the latter, since this transformation can increase runtimes by several orders of magnitude.

\end{abstract}

\begin{keywords}
  Vertex Weighted Matching, Graph Algorithms,  Approximation Algorithms
\end{keywords}

\begin{AMS}
  68Q25, 68R10, 68U05
\end{AMS}

\section{Introduction}
\label{section:introduction}

We consider a variant of the matching problem in graphs in which weights are assigned to the 
vertices,  the weight of a matching is the sum of the weights on the matched vertices,
and we are required  to compute a matching of maximum weight. 
We call this the maximum vertex-weighted matching problem (MVM). 
In this paper we describe a $2/3$-approximation algorithm for MVM in 
bipartite graphs and  implement it efficiently. 
We compare its performance with several algorithms:
an algorithm for computing maximum edge-weighted matchings, an (exact) algorithm for MVM, a $1/2$-approximation Greedy algorithm, and a $(1-\epsilon)$-approximation algorithm for maximum edge weighted matchings. 

Matching is a combinatorial  problem that has been extensively studied since the 1960's. 
The problem of computing a matching with the maximum cardinality of edges is 
the maximum cardinality matching (MCM) problem.
When weights are assigned to the edges, the problem of computing a matching with 
the maximum sum of weights of the matched edges is  the maximum edge-weighted matching 
problem (MEM). There are other variants too: One could ask for 
a matching that has the maximum or minimum sum of edge weights
among all  maximum cardinality matchings. 
Or one could ask for a maximum bottleneck matching, where we seek to maximize the minimum weight 
among all matched edges. 
While there have been a number of papers on the MCM and MEM problems,
there has been little prior work on the MVM problem. 
Let $G= (V, E)$ denote a bipartite graph with $n \equiv |V|$ vertices and $m \equiv |E|$ edges. 
Spencer and Mayr~\cite{spencer1} have described an algorithm to solve the MVM exactly with  $O(m \sqrt{n}  \log n)$ time complexity.
We do not know of prior work on approximation algorithms for 
the MVM problem.  
Background information on matchings and approximation algorithms is provided in the next Section. 

For the MCM problem on bipartite graphs, it is now known that algorithms with  $O(nm)$ worst-case time complexity are among the practically fastest algorithms relative to asymptotically faster algorithms, e.g., the Hopcroft-Karp algorithm  with $O(m \sqrt{n})$ complexity. Among these algorithms are a Multiple-Source BFS-based algorithm, a  Multiple-Source DFS-based algorithm,   and a Push-Relabel algorithm 
\cite{Azad+:TPDS,Duff+:2011,PF90}.
All of these algorithms  employ  greedy initialization algorithms such as the Karp-Sipser algorithm, 
and  employ other  enhancements to make the implementations run fast.  
These algorithms have also been implemented in parallel  on modern shared-memory multithreaded processors. 
 
An MVM problem can be  transformed into an MEM problem by summing the weights 
of an endpoints of an edge and assigning the sum as the weight of an edge. 
Thus an exact or approximation algorithm for MEM becomes an exact or approximation
algorithm for MVM as well. 
\chg{However, we show in the experimental section that when edge weights are 
derived this way, they are highly correlated, adversely affecting the run times of the algorithms.}
It is possible to design exact and approximation algorithms that do not make use of linear programming formulations for MVM. 
These  algorithms are conceptually simpler, 
easier to implement, and  provide insights into  the ``structure'' of the MVM problem. 
The $2/3$- and $1/2$-approximation algorithms have near-linear time complexity with small constants, 
and in practice, are quite fast and deliver high quality approximations. 
\chg{A  $2/3$-approximation algorithm for the MVM problem is obtained by limiting the augmenting path lengths to $3$, 
but this technique does not lead to a $2/3$-approximation for 
MEM.}


MVM problems arise in many contexts, such as the design of 
network switches~\cite{tabatabaee1},
schedules for training of astronauts~\cite{bell1}, 
computation of sparse bases for the null space or the column space of a rectangular
matrix~\cite{coleman3,pinar1}, etc. 
Our interest in this problem was sparked by our work on sparse bases for the null space and column space of rectangular matrices. The null space basis is useful in successive quadratic programming (SQP) methods to solve nonlinear optimization problems; here the matrix sizes 
are determined by the number of constraints in the 
optimization problem. 
In this context a matroid greedy algorithm can be shown to compute 
a sparsest such basis. For the null space basis, the problem still remains 
NP-hard since computing a sparsest null vector is already NP-hard. 
However, for the column space basis, this leads to a polynomial time 
algorithm, and it can be efficiently implemented by computing  a maximum weight vertex-weighted matching. 
As will be seen from our results, approximation algorithms are needed to make this algorithm practical since an optimal algorithm can take several hours on large graphs. 
We suspect that a number of problems that could  be modeled as MVM problems 
have been modeled in earlier work as MEM problems due to the extensive literature on the latter problem. 

The remainder of this paper is organized as follows. 
Section~\ref{sec:sectionBackground} provides background  on matching,
and describes the concepts of $M$-reversing and $M$-increasing paths in a graph. 
The next Section~\ref{sec:sectionCharacterizationOptimal} characterizes 
MVMs using the concept of weight vectors as well as augmenting paths and increasing paths. 
An exact algorithm for MVM is briefly described in 
Section~\ref{sec:sectionOptimalSolutions}, 
and Section~\ref{sec:sectionTwoThirdApprox} describes a  $2/3$-approximation algorithm for MVM on bipartite graphs.
Next, Section~\ref{sec:sectionCorrectness-TwoThirdsAlg} proves the correctness 
of the $2/3$-approximation algorithm.
Section~\ref{sec:sectionHalfApprox} briefly discusses the Greedy $1/2$-approximation
algorithm and its correctness. 
Computational results on the the exact MEM  and MVM algorithms, the Greedy $1/2$-approximation,  $2/3$-approximation, and scaling-based $(1- \epsilon)$-approximation algorithms are included in Section~\ref{sec:sectionExperiments}. 
We conclude in the  final Section~\ref{sec:sectionConclusions}.  

Preliminary versions of these results were included in the PhD thesis 
of one of the authors~\cite{Halappanavar:thesis},
and in an  unpublished report~\cite{Dobrian+:report}.

\section{Background}
\label{sec:sectionBackground} 

We define the basic terms we use here, and refer the reader to discussions in the following 
books for additional background on matching theory~\cite{Burkard1,lawler1,lovasz1,papadimitriou1,schrijver1}, and approximation algorithms~\cite{vazirani1,Williamson+:book}.  

A {\em matching\/}  in a graph $G= (V, E)$ is a set of edges $M$ such that 
for each vertex $v \in V$, at most one edge in $M$ is incident on $v$. 
An edge in $M$ is a matched edge, and otherwise, it is an unmatched edge. 
Similarly a vertex which is an endpoint of an edge in $M$ is matched, 
and otherwise it is an unmatched vertex. 

A {\em path\/} in a graph is a sequence of distinct vertices  $\{u_i: i=1, \ldots, k\}$ such that 
consecutive vertices form an edge in the graph.   
The length of a path is the number of edges (not the number of vertices) in it. 
A cycle is a path in which the first and last vertices are the same. 
Given a matching $M$ in a graph, an {\em $M$-alternating path\/} 
is a path in which matched and unmatched edges alternate. 
If the first and the last vertices  in an $M$-alternating path $P$ are unmatched,
then it is an {\em $M$-augmenting path\/}, since by flipping  the matched 
and unmatched edges along $P$ we obtain a new matching that has one more edge than the original matching. 
The augmented matching is obtained by the symmetric difference $M \oplus P$. 
An augmenting path must have an odd number of edges since it has one more unmatched
edge than the number of  matched edges. 
Figure~\ref{Fig:no-incr-paths} shows an augmenting path joining vertices $u$ and $v$ in the top figure of Subfigures (1) and (2). 
Here solid edges are matched, and the dashed edges are unmatched.
(The vertices $w$, $w'$, and the neighbor of $w'$ are not involved in this path.)
The results of an augmentation are shown in the bottom figures
of these Subfigures. 

An {\em $M$-reversing path\/}  is an alternating path with one endpoint $M$-matched and 
the other endpoint $M$-unmatched.  
Such a path has even length (number of edges), and half the edges are matched and 
the other half are unmatched. 
In Figure~\ref{Fig:no-incr-paths}, the paths joining $w$ and $w'$ in the bottom figures of Subfigures (1) and (2) are 
$M'$-reversing paths. 
We can exchange the matched and unmatched edges on a reversing path without changing the cardinality of the matching. 
However, it might be possible to increase the weight of a vertex-weighted matching in this way. 
An {\em $M$-increasing path\/}  is an $M$-reversing path such that its $M$-unmatched endpoint  $u$ 
is heavier than its $M$-matched endpoint $u'$.  
Let  $\phi(u)$ denote the vertex weight of a vertex $u$. 
For an increasing path, we have $\phi(u) > \phi(u')$. If we exchange matched and 
unmatched edges along this path, then the weight of the matching increases by $\phi(u) - \phi(u')$. 

An {\em exact algorithm\/} for a maximization version of an optimization problem on graphs computes a 
solution with the maximum value of its objective function. 
For the matching problems considered here,
there exist polynomial time algorithms for computing a 
maximum weighted matching. However, since the time complexity
of these algorithms is high, recent work has focused on 
developing approximation algorithms for these problems that
run in nearly linear time in the number of edges in the graph. 
An {\em approximation algorithm\/} for a maximization problem
computes a solution such that the ratio of the 
value of the objective function obtained by the approximation
algorithm to that of the exact algorithm is bounded by a constant or a function of the input size, for all graphs
that could be input to the problem. An upper bound on this 
ratio over all graphs is the approximation ratio of the algorithm. 
The approximation algorithms that we design  for the MVM problem  satisfy approximation ratios of $2/3$ or $1/2$. 

Hopcroft and Karp~\cite{hopcroft1} showed that if $M$ is a matching in a graph $G$ such that a shortest augmenting path has length at least $(2k-1)$ edges, then  $M$ is a $(k-1)/k$-approximation matching for a maximum cardinality matching.

Recent work has focused on  developing several approximation algorithms for MEM
that run in  time linear in the number of edges in the graph. 
A number of  $1/2$-approximation algorithms are known for MEM,
including the Greedy algorithm, the Locally Dominant edge algorithm, and a Path-growing algorithm~\cite{preis1, drake1}. 
Currently the practically fastest $1/2$-approximation algorithm is the Suitor algorithm of Manne and Halappanavar~\cite{MH14}, 
which employs a proposal based approach similar to algorithms for the  stable matching problem.  
This algorithm has a higher degree of concurrency since vertices can be processed in any order to extend proposals to their eligible heaviest  neighbors, since proposals can be annulled. The parallel Suitor algorithm has $O(\log m \log \Delta)$ parallel depth and $O(m)$ work when the edge weights are chosen uniformly at random (here $\Delta$ is the maximum degree of a vertex)~\cite{Khan+:parallel-bedgecover}. 
The Locally Dominant edge algorithm and the Suitor algorithm have also been implemented on 
multi-threaded parallel architectures~\cite{Halappanavar+:Computer,MH14}. 
 
In practice, some of the $1/2$-approximation algorithms compute matchings with $95\%$ or more of the weight of a maximum edge-weighted matching for many graphs. 
In addition, $(2/3 - \epsilon)$-approximation algorithms have also been designed for MEM
~\cite{drake2,pettie1}. 
These algorithms are slower than the $1/2$-approximation algorithms and do not 
improve the weight of the matching by much in practice~\cite{MS07}. 

More recently, for any $\epsilon >  0$,  
a  $(1-\epsilon)$-approximation algorithm for MEM with time complexity 
$O(m \epsilon^{-1} \log \epsilon^{-1})$ has been designed
by Duan and Pettie~\cite{Duan+:jacm}. 
This algorithm is based on a scaling-based 
primal-dual  approach, requires the computation and updating of blossoms for non-bipartite graphs, and is  more expensive than   
the simpler $1/2$-approximation algorithms.
We will show in the Results section that this algorithm is slower than the $2/3$-approximation algorithm considered in this paper, 
while surprisingly computing matchings of lower weight.  
The Duan-Pettie  paper  surveys earlier work on exact and approximation algorithms for the MEM problem. 
Hougardy~\cite{Hougardy:survey}  has also provided a recent survey of developments in 
approximation algorithms for matchings. 
MEM problems arise in sparse matrix computations (permuting large elements to the diagonal of a sparse matrix)~\cite{duff1}, network alignment~\cite{Khan+:SC12}, scheduling problems, etc.

Approximation algorithms  have  now been designed for  several problems related to matching: 
maximum vertex-weighted matching, maximum edge-weighted matching, maximum edge-weighted $b$-matching
\cite{Khan+:b-matching,Khan+:SC12}, the minimum weight edge cover, and the minimum weight $b$-edge cover problem
\cite{Khan+:b-edgecover}. Approximation is a paradigm for designing parallel algorithms for these problems, and such algorithms has been shown to have good parallel performance.

\section{Characterization of Maximum  Vertex Weighted Matchings}
\label{sec:sectionCharacterizationOptimal}

In this section we characterize an  MVM two different ways:
First,   in terms of augmenting paths and increasing paths, and second, in terms of the weights in the  matching. 

If all the vertex weights are positive, then any maximum vertex weighted matching is 
a maximum cardinality matching as well. 
If some of the vertex weights are zero, then without loss of generality,
we can  choose a maximum vertex weighted matching to have maximum cardinality also as shown below.  
\begin{Lemma}
Let $G = (V, E)$ be a graph and $\phi: V \mapsto R_{\geq 0}$ be a non-negative weight function. 
There is a maximum vertex-weighted  matching $M$ that is also a maximum cardinality matching in $G$. 
\end{Lemma}

\begin{proof}
Consider what happens to the vertices when we  augment a vertex-weighted matching by an augmenting path. 
Both endpoints of the augmenting path are now matched (these were previously unmatched), and all interior vertices in the path continue to remain matched. (Figure~\ref{Fig:no-incr-paths} illustrates this.) 
Thus in an algorithm that computes vertex-weighted matchings solely by augmentations, once a vertex is matched it is never  unmatched, 
and it will be matched at every future step in the algorithm.
(We call this the ``once a matched vertex, always a  matched vertex'' property of augmentations of a vertex weighted matching.) 
This implies that if the weights are non-negative, each augmentation causes the weight of a matching to increase or stay the same. Thus we can always choose an MVM to have maximum cardinality of edges. 
\end{proof}

Of course, the set of matched edges and unmatched edges are exchanged along an augmenting path, so there is no  corresponding ``once a  matched edge, always a  matched edge'' property.
Note also that when we use an increasing path between two vertices $w$ and $w'$ to increase the weight of a matching, 
then the vertex $w$ gets matched, $w'$ gets unmatched, 
and all interior vertices in the path continue to be matched. 
(Again, Figure~\ref{Fig:no-incr-paths} provides examples.)
Hence the property of ``once a matched vertex, always a matched vertex" is not true of  an algorithm that uses increasing paths during its execution. 

If some of the vertex weights are negative, we can transform the problem so that we need
consider only nonnegative weights  as shown in Spencer and Mayr~\cite{spencer1}.
For each vertex $v$ with a negative weight, we add a new vertex $v'$, an edge $(v,v')$,
with  the weight of $v'$ set to the absolute value of the weight of $v$,
and the new weight of $v$ set to zero.  
An MVM in the transformed graph leads to an MVM in the original graph; 
however, this transformation might not preserve approximations.  
From now on, we assume that all weights are non-negative.

We turn to the first of our characterizations of a MVM. 
\begin{Theorem}
Let $G = (V, E)$ be a graph and $\phi: V \mapsto R_{\geq 0}$ be a non-negative weight function.  
A  matching $M$ is an MVM that also has maximum cardinality if and only if 
(1) there is no $M$-augmenting path in $G$, and 
(2) there is no $M$-increasing path in $G$. 
\end{Theorem} 

\begin{proof}
A matching $M$ has maximum cardinality if and only if there is no augmenting path with respect to it~\cite{papadimitriou1}. Hence we need to prove only (2). 

For the  {\em only if\/} part, 
if there were an $M$-increasing path $P$, then the symmetric difference $M \oplus P$ would yield a 
vertex-weighted matching of larger weight, contradicting the assumption that $M$ has maximum vertex weight. 
For the {\em if\/} part, consider a maximum vertex-weighted matching $M_1$ and 
a matching that does not have an augmenting path or increasing path with respect to it, $M_2$. 
We will show that $M_2$ has the same weight as $M_1$. 
The symmetric difference $M_1 \oplus M_2$ consists of cycles and paths. 
A cycle consists of vertices matched by both matchings, and hence cannot account for any difference between
them in weight. 
Every path must have even length, and an equal number of edges from $M_1$ and $M_2$, 
for otherwise we would be able to augment one of the two matchings, and we  need  consider only increasing paths. 
By our assumption, $M_2$ does not have an increasing path with respect to it. 
But there cannot be an increasing path $P$ with respect to $M_1$ either, 
for such a path would enable us to increase its weight
by the symmetric difference $M_1 \oplus P$. 
\end{proof}
This Theorem was proved  by Tabatabaee et al.~\cite{tabatabaee1},
who  seem to restrict the result to  bipartite graphs. 

Now we characterize an MWM in terms of the weights. 
Given a matching $M$ we can define a weight vector $\Phi(M)$ that lists the weights of all the 
vertices matched by $M$ in non-increasing order. A weight is listed multiple times if there is more than one vertex with the same weight.  
We can compare the weight vectors of two matchings $M_1$ and $M_2$ in a graph $G$ lexicographically:  
we define the weight vector $\Phi(M_1) > \Phi(M_2)$ if $w_1 > w_2$,
where $w_1$ is the \emph{first} value in $\Phi(M_1)$  not equal to the corresponding value $w_2$ in $\Phi(M_2)$.  This definition can compare  matchings of different sizes in the same graph, since the matching with fewer edges can be augmented with zeros. 


\begin{Theorem}
Let $G = (V, E)$ be a graph and $\phi: V \mapsto R_{\geq 0}$ be a non-negative weight function.  
A matching $M$ is an MVM   if and only if its  weight vector 
$\Phi(M)$ is lexicographically maximum among all weight vectors. 
\end{Theorem} 

\begin{proof}
For the {\em only if\/} part, let $M$ denote a maximum vertex-weighted matching and $M^l$ denote 
a matching whose weight vector is lexicographically maximum. 
By our choice, the matching $M$ has maximum cardinality.  
Similarly, $M^l$ also has maximum cardinality, for otherwise we could augment the matching 
to a maximum cardinality matching while keeping all of the matched vertices in $M^l$ matched
in the augmented matching,
due to the once-matched, always-matched property of augmentations. 
Hence suppose that the matching $M^l$ has weight less than the weight of $M$, and that 
the weight vector $\Phi(M)$ is not lexicographically maximum. 

Let the first lexicographic difference between the vectors $\Phi(M^l)$ and $\Phi(M)$ correspond 
to a vertex $u$ that is matched in  $M^l$ and unmatched in $M$. 
Now consider the symmetric difference of the two matchings $M \oplus M^l$. 
Since both matchings have the same maximum cardinality, the symmetric difference consists of 
cycles or  paths of even length in which edges from the two matchings alternate. 
As stated earlier, a cycle cannot contribute to the 
difference in the weights between the matchings. 
Among the alternating paths, there is one path $P$ of even length whose one endpoint is the 
vertex $u$ that is matched in $M^l$ but not in $M$. Denote the other endpoint of this path by $u'$. 
Since the path has even length, $u'$ is matched in $M$ but not in $M^l$. 
Also since the first lexicographic difference between the vectors $\Phi(M^l)$ and 
$\Phi(M)$ occurs at $u$,  and $\Phi(M^l)$ but not $\Phi(M)$ is lexicographically maximum, 
the weight of $u$ is greater than or equal to the weight of $u'$. 
The matching $M \oplus P$ would increase the weight of the maximum vertex-weight matching $M$
if the two weights were unequal. Hence these two weights are equal,
and we obtain a contradiction to our assumption that this was the first weight
where the weight vectors of the two matchings were different.  


The proof of  the {\em if\/} part is in \cite{mulmuley1}, and 
we include it here for completeness. 
Again let $M$ denote a maximum vertex-weighted matching, and let $M^l$ denote a matching 
whose weight vector is lexicographically maximum.
The symmetric difference $M \oplus M^l$ consists of cycles and paths. 
As stated earlier,  vertices  in alternating cycles cannot contribute to the differences in the weight vectors. 
Now the matching $M^l$ must have maximum cardinality since otherwise we could 
augment it  and get a lexicographically larger weight vector. 
Since both matchings $M$ and $M^l$ have maximum cardinality, there is no 
augmenting path with respect to either matching. 
Hence each path in the symmetric difference $M \oplus M^l$  must have even length. 
Let $u$ be one endpoint of one such  path $P$ that is matched in $M^l$ and unmatched in $M$, 
and $u'$ denote the other endpoint of $P$ that is matched in $M$ and unmatched in $M^l$. 
Since $M^l$ has the  lexicographically maximum  weight vector, 
we can only have $\phi(u') < \phi(u)$. 
Hence this is an increasing path and by replacing  the matched edges in $M$ 
on the path $P$ by  the edges in $M^l$  on $P$, we could increase the weight of the 
maximum vertex-weighted matching $M$. 
This contradiction proves the result. 
\end{proof}

The structural properties in these results facilitate the design of two classes of algorithms 
for MVM. 
One approach is to compute a maximum weighted matching from an empty matching, augmenting the 
matching by one edge in each iteration of the algorithm. By choosing to match 
vertices $u$ in decreasing order of weights, and by choosing a heaviest unmatched vertex 
reachable from $u$, we can ensure that an increasing path with 
respect to the current partial matching does not exist in the graph.
We call this the direct approach. 
The second, speculative approach, would begin with any  matching of maximum cardinality, and 
increase the weight by means of increasing paths, 
until a matching of maximum weight is reached. 

There are advantages associated with each of these approaches.
The direct approach, together with recursion,   has been employed  by 
Spencer and Mayr~\cite{spencer1} to design an $O(n^{1/2} m \log n)$ algorithm for MVM.  
The speculative approach could be efficient in combination with the 
Gallai-Edmonds decomposition~\cite{lovasz1}. 
This decomposition identifies a subgraph which has a perfect matching
in any maximum cardinality matching; in such a subgraph, 
{\em any\/}  maximum cardinality matching is a  maximum vertex-weighted matching as well. Thus we need solve an MVM only in the remainder of the 
graph. 
If the subgraph with the perfect matching is large, there could be 
substantial savings in run-time. 
Since an MCM can be computed practically much faster than an MVM, and the Gallai-Edmonds decomposition can be obtained in linear time from an MCM, this approach might be practically useful.

\section{An Exact  Algorithm for MVM}
\label{sec:sectionOptimalSolutions}

In this Section, we describe an algorithm that solves the MVM problem exactly,
primarily to show how our $2/3$-approximation algorithm can be derived from it in a natural manner. 
In Algorithm  MATCHD (see the displayed Algorithm~\ref{Exact-General-VWM}, 
we describe how an  MVM is computed by matching  
vertices in non-increasing order of weights. 
Here $Q$ is the set of unmatched vertices, and in each iteration the algorithm attempts to match 
a heaviest  unmatched  vertex $u$. 
From $u$, the algorithm searches for a heaviest unmatched vertex $v$ it can reach by 
an augmenting path $P$. If it finds $P$, then the matching is augmented by forming 
the symmetric difference of the current matching $M$ with $P$, and the vertices $u$ and $v$ 
are removed from the set of unmatched vertices. 
If it fails to find an augmenting path from $u$, then 
$u$ is removed from the set of unmatched vertices, since we do not need to search 
for an augmenting path from $u$ again. 
When all the unmatched vertices have been processed, the algorithm terminates. 

\begin{algorithm}
\centering
\caption{\textbf{Input:} A  graph $G$ with weights $\phi(V)$  on the vertex set $V$. 
\textbf{Output:}  A vertex-weighted matching  $M$ of maximum weight. 
\textbf{Effect:} Computes a maximum vertex-weight matching in the  graph $G$.}
\label{Exact-General-VWM}
\begin{algorithmic}[1]
\Procedure{\textsc{MatchD}}{$G = (V, E),  \phi(V)$} 
	\State $M \gets \phi;$ 
	\State $Q \gets V;$
	\While {$ Q  \neq \emptyset$}  
		\State $u \gets heaviest(Q);$ 
		\State $Q \gets Q - u;$ 
		\State Find an  augmenting path $P$  from  $u$ that reaches  
                         \newline \indent \indent \indent a heaviest unmatched vertex $v$;  
		\If{$P$ $\mathit{found}$} 
			\State $M \gets M \oplus P;$
                         \State $ Q \gets Q - v$; 
		\EndIf
	\EndWhile 
\EndProcedure
\end{algorithmic}
\end{algorithm}

To prove the correctness of the algorithm, we need the following Lemma. 

\begin{figure}[!thp]
\centering
\includegraphics[scale=0.1]{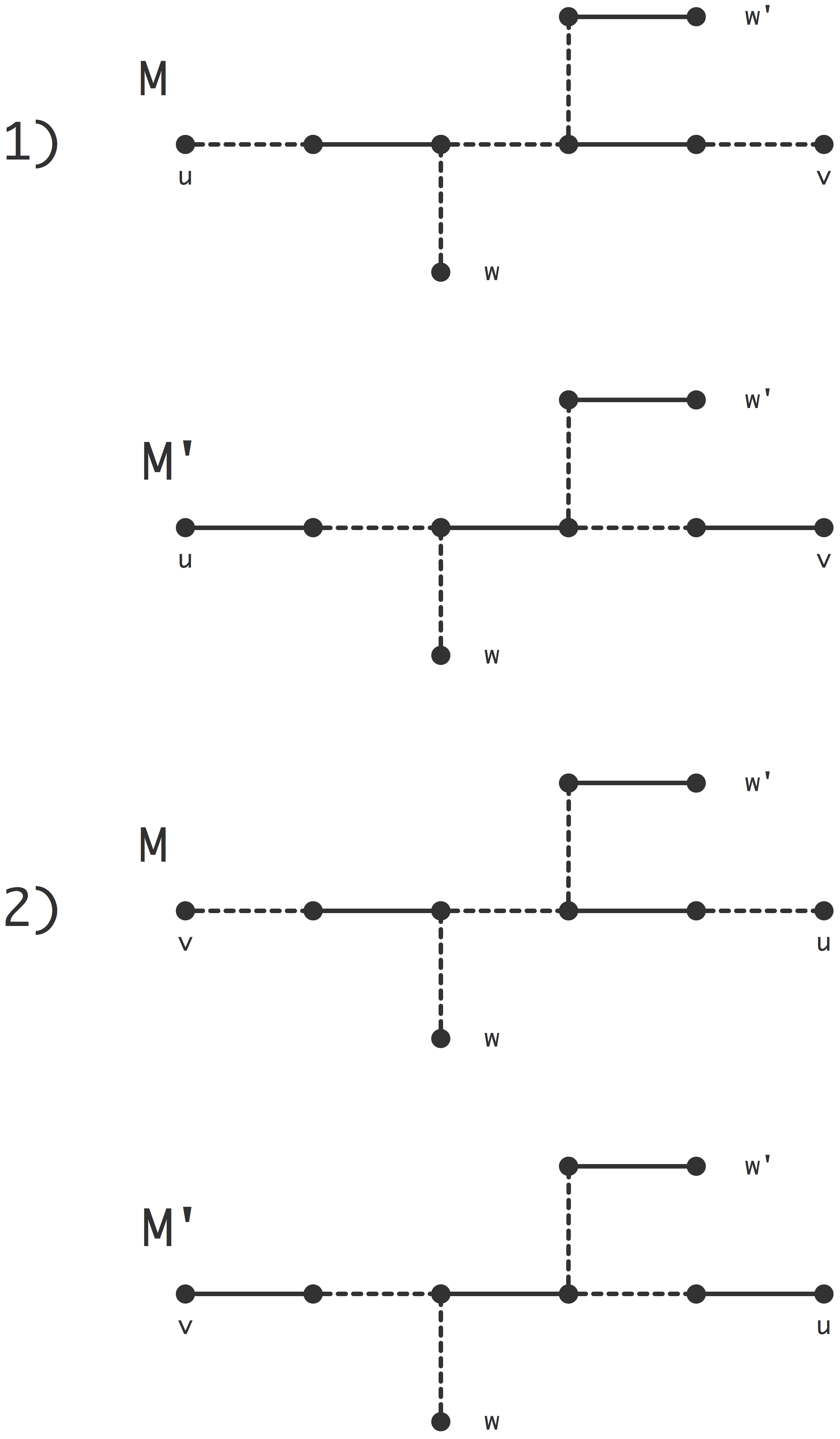}
\caption{\textit{Construction used in the proof of Lemma~\ref{lem:no-aug-incr-paths}}. 
In each case, $M$ is a matching, $P$ is an $M$-augmenting path between 
$u$, a heaviest $M$-unmatched vertex and $v$, a heaviest unmatched vertex reachable 
by an $M$-alternating path from $u$,
and $M' = M \oplus P$. 
Matched edges are drawn as solid edges, and unmatched edges are drawn as dashed edges. 
If there is no $M$-increasing  path in $G$, 
then there is no $M'$-increasing path as well. 
}
\label{Fig:no-incr-paths}
\end{figure}

\begin{lemma} 
\label{lem:no-aug-incr-paths}
Let $z$ be an unmatched vertex with respect to a matching $M$ in a graph $G=(V,E)$, and let 
$\phi: V \mapsto R_{\geq 0}$ be a weight function on the vertex set $V$. 
Suppose that there does not exist an $M$-augmenting path from the vertex $z$, 
and also that there is no $M$-increasing path (from any vertex) in the graph $G$. 
Let $P$ be an $M$-augmenting path from a heaviest unmatched vertex $u$, 
whose other endpoint $v$ is a heaviest unmatched vertex that can be reached from $u$ 
by an $M$-alternating path. 
If $M' = M \oplus P$, then there does not exist an $M'$-augmenting path from the vertex $z$,
nor an $M'$-increasing path (from any vertex) in the graph $G$.  
\end{lemma} 

\begin{proof}
When $P$ is an augmenting path from 
some $M$-unmatched vertex $u$, clearly $u$ has to be distinct from  
the vertex $z$ since from the latter, there is no augmenting path by the condition of the Lemma.  
A proof that there is no $M'$-augmenting path from $z$ can be found in \cite{papadimitriou1}. 
Hence we prove that there is no $M'$-increasing path in $G$. 
(Similar arguments will be made several times in this paper.) 

If there is no $M'$-reversing  path in $G$, then there cannot be any $M'$-increasing path, and we are done.
Hence choose an arbitrary $M'$-reversing path $P'$ that joins an $M'$-unmatched vertex $w$ and 
an $M'$-matched vertex $w'$.
Since every vertex on the $M$-augmenting path $P$ is matched in $M'$, 
the vertex $w$ cannot belong to $P$, while the vertex $w'$ can belong to $P$ and does not need to be 
distinct from the vertices $u$ or  $v$. 
We will prove that $\phi(w) \leq \phi(w')$ and hence that the  path $P'$ is not $M'$-increasing. 

If an $M$-reversing path also  joins the vertices $w$ and $w'$, where $w$ is  $M$-unmatched and $w'$ 
is $M$-matched, then since there is no $M$-increasing path in $G$, we have $\phi(w) \leq \phi(w')$. 
If no $M$-reversing path joins $w$ and $w'$, then the paths $P'$ and $P$ cannot be vertex-disjoint;
for if they were, then $P'$ would also be an $M$-reversing path, which we assumed does not exist in $G$.
Thus the paths  $P$ and $P'$ share at least one common vertex, and indeed, as we show now,
it shares a matched edge.
For, every vertex on the path $P$ is $M'$-matched, and hence a vertex in $x$ in $P' \setminus P$ 
that is adjacent to a vertex $y$ in $P$ must have the edge $(x,y)$ as an $M'$-unmatched edge. 
Since $P'$ is an $M'$-alternating path, the next edge on the path $P'$ must be a matched 
edge incident on the vertex $y$, and hence this matched edge is common  to both paths $P'$ and $P$. 
(The paths $P$ and $P'$ could intersect more than once.)

Now we have two cases to consider. 

The cases are illustrated in 
Fig.~\ref{Fig:no-incr-paths}. 
In the first case, there is an $M$-augmenting path between $u$ and $w$,
and there are two subcases: either $v$ and $w'$ are the same vertex,
or there is an $M$-reversing path $Q$ between $v$ and $w'$. 
The second subcase  corresponds to Subfigure (1). 
Now the path $Q$ cannot be an $M$-increasing path by our assumption that no such path
exists in $G$. 
Hence in both subcases, we can write $\phi(v) \leq \phi(w')$. 
Since we chose the path $P$ to begin at $u$ and end at the $M$-unmatched vertex $v$ and not at the 
$M$-unmatched vertex $w$, we have $\phi(w) \leq \phi(v)$. 
Combining the two inequalities, we obtain $\phi(w) \leq \phi(w')$. 

In the second case, there is an $M$-augmenting path between $v$ and $w$,
and again there are two subcases:
either $u$ and $w'$ are the same vertex,  or there is an $M$-reversing path $Q'$ 
between $u$ and $w'$. The second subcase is illustrated in Fig.~\ref{Fig:no-incr-paths} (2). 
As before, the path $Q'$ cannot be $M$-increasing by supposition, and therefore 
$\phi(u) \leq \phi(w')$. 
Since $u$ is a heaviest $M$-unmatched vertex by choice, and $w$ is $M$-unmatched, we have 
$\phi(w) \leq \phi(u)$. 
Combining, we have $\phi(w) \leq \phi(w')$. 
\end{proof}


\begin{theorem}
\label{thm:optimal}
Algorithm MatchD computes an MVM  in a graph $G=(V,E)$ with 
vertex weights given by a function $\phi: V \mapsto R_{\geq 0}$. 
\end{theorem}

\begin{proof}
Let $M$ be the matching computed by Algorithm MatchD. 
We show by induction that there does not exist an $M$-augmenting path nor an $M$-increasing path 
in the graph $G$. 

Let $n_a$ be the number of augmenting operations in the Algorithm MatchD.  
The matching $M$ is the last in a sequence of matchings $M_i$, for $i=0$, $1$, $\ldots$, $n_a$,
computed by the algorithm. 
For $0 \leq i < n_a$, let $P_i$ denote the $M_i$-augmenting path used to augment $M_i$ to the 
matching $M_{i+1}$, and let $u_i$ denote the source of the augmenting path (the $M_i$-unmatched 
vertex from which we searched for an augmenting path), and let $v_i$ denote its other end point. 
The induction is on the matching $M_i$, and the inductive claim is that 
\newline (1) there is no $M_i$-augmenting path from an unmatched  vertex that has already been 
processed,  i.e., a vertex from which we have searched for an augmenting path 
earlier and have failed to find one, and 
\newline (2) there is no $M_i$-increasing path  from any vertex in  $G$. 

The basis of the induction is $i=0$, when the result is trivially true. The first condition holds 
because no vertices have been processed yet, and the second condition holds since the matching is
empty and hence there is no  increasing path. 
Hence assume that the claim is true for some $i$, with $0 \leq  i < n_a$. 
Now the result holds for the step $i+1$ by applying Lemma~\ref{lem:no-aug-incr-paths}.
\end{proof}

The time complexity of this algorithm is $O(nm + n \log n)$.
We seek to match each vertex, and the search for augmenting paths from each vertex costs 
$O(m)$ time. The second term is the cost of sorting the vertex weights. 
 
Additionally, we can describe an exact algorithm for MVM that takes the speculative approach.
Here one computes first a maximum cardinality matching, and then 
searches for increasing paths from unmatched vertices, in decreasing order of 
weights, to obtain an MVM. 
We need additional results to show that this algorithm computes an MVM.
The time complexity of the algorithm is the same as the one using the direct approach
described in this Section. 
Practically, the performance of the two classes of algorithms could be quite different,
and hence it is worthwhile to implement these algorithms.
However, since our interest in this paper is on a $2/3$-approximation algorithm for MVM in
bipartite graphs, we do not discuss this further here.

\section{A $2/3$-Approximation Algorithm for MVM in Bipartite Graphs}
\label{sec:sectionTwoThirdApprox}

In this Section, we restrict ourselves to bipartite graphs. 
In order to solve the MVM on a bipartite graph $G=(S,T, E, \Phi)$, we create two 
`one-side weighted'  subproblems from the given problem.
In the first subproblem, the weights on the $T$ vertices are set to zero,
and in the second subproblem, the weights on the $S$ vertices are set to zero. 
We compute   MVMs on the two subproblems, and then combine them,
using the Mendelsohn-Dulmage theorem, to obtain a solution of the original problem.  
In this section, we describe the algorithm, and compute its time complexity. We defer the proof of correctness of the algorithm to the next section, since it is somewhat lengthy.

\begin{Theorem}[Mendelsohn-Dulmage]~\cite{Mendelsohn+:theorem}
Let $G= (S, T, E)$ be a bipartite graph, and let $M_1$ and $M_2$ be two matchings in $G$.
Then there is a matching $M \subseteq M_1  \cup M_2$ such that all $M_1$-matched
vertices in $S$ are matched in $M$, and all $M_2$-matched vertices in $T$ are also matched in $M$. 
\end{Theorem} 

The matching $M$ is obtained by a case analysis that considers the symmetric difference 
of $M_1$ and $M_2$, and a proof is included in Section $5.4$ of Lawler~(\cite{lawler1}).

\subsection{The Approximation Algorithm} 

The Approximation Algorithm (displayed in Algorithm~\ref{VWMAB3})
calls a  Restricted Bipartite Matching algorithm
(in turn displayed in Algorithm~\ref{Restr-Bipart-Approx})
which solves a one-side weighted MVM in a bipartite graph. 
The latter algorithm matches  unmatched vertices (in the weighted vertex part)
in decreasing order of weights.
From each unmatched vertex $u$, the algorithm searches for  an  unmatched vertex 
(it is  unweighted) by a shortest  augmenting path of \emph{length at most three}. 
If it finds a short augmenting path, then the matching is augmented by the path;
if it fails to find such a path, then we do not consider the vertex $u$ again in the algorithm. 

After solving the two Restricted Bipartite Matching problems, the algorithm  invokes the 
Mendelsohn-Dulmage theorem to compute a  final matching in which 
the matched vertices from the weighted part of each problem  are included.
We will prove that this algorithm computes a $2/3$-approximation to the 
MVM, and that it can be implemented in $O(n \log  n + m )$ time.

\begin{algorithm}
\centering
\caption{\textbf{Input:} A bipartite graph $G$ with weights $\phi$ on the vertices. 
\textbf{Output:} A matching $M$. 
\textbf{Effect:} Computes a $\frac{2}{3}$-approximation to a maximum vertex-weighted matching.}\label{VWMAB3}
\begin{algorithmic}[1]
\Procedure{\textsc{Bipartite-TwoThird-Approx}}{$G=(S,T,E), \phi:S \cup T  \rightarrow \mathbf{R}_{\geq 0}$.}
	\State $M_{S} \gets \textsc{Restricted-Bipartite-Match}(G, S, \phi(S));$ 
	\State $M_{T} \gets \textsc{Restricted-Bipartite-Match}(G, T, \phi(T));$ 
	\State$M \gets $\textsc{MendelsohnDulmage}$(M_{S},M_{T},M);$ 
\EndProcedure
\end{algorithmic}
\label{alg:approx}
\end{algorithm}

\begin{algorithm}
\centering
\caption{\textbf{Input:} A bipartite graph $G$ with weights $\phi(S)$ only on one vertex part $S$. 
\textbf{Output:} A matching $M_S$. 
\textbf{Effect:} Computes a $\frac{2}{3}$-approx to a maximum vertex-weight matching in a bipartite 
graph.}
\label{Restr-Bipart-Approx}
\begin{algorithmic}[1]
\Procedure{\textsc{Restricted-Bipartite-Match}}{$G, S, \phi(S)$} 
	\State $M_{S} \gets \phi;$ 
	\State $Q \gets S;$
	\While {$ Q  \neq \emptyset$}  \Comment{Compute $M_{S}$}
		\State $u \gets heaviest(Q);$ 
		\State $Q \gets Q - u;$ 
		\State Find a shortest augmenting path $P$ of length at most $3$ starting at $u;$ 
		\If{$P$ $\mathit{found}$} 
			\State $M_{S} \gets M_{S} \oplus P;$
		\EndIf
	\EndWhile 
\EndProcedure
\end{algorithmic}
\end{algorithm}


\subsection{Time Complexity of the $2/3$-Approximation Algorithm}

\begin{Theorem} 
The $2/3$-approximation  algorithm has time complexity $O(n \log n \, + \, m)$, where $n$ is the maximum of $|S|$ and $|T|$,
and $m$ is the number of edges in the  bipartite graph $(S, T, E)$. 
\end{Theorem} 

\begin{proof}
We will establish the time complexity for the restricted bipartite graph with nonzero weights 
on the $S$ vertices. 
An identical  result holds for the graph with nonzero weights on the  $T$ vertices.
The cost of computing the final matching via the Mendelsohn-Dulmage theorem is $O(n)$,
since it needs to work with only the symmetric difference of the two matchings. 
The $n \log n$ complexity comes from the sorting of the weights on the vertices in decreasing 
order.

In each iteration of the {\bf while} loop, we choose an unmatched vertex $s$, and examine all neighbors 
of $s$. If we find an unmatched vertex $t$, then we can match the edge $(s,t)$ and we proceed 
to the next iteration. In this case, when an augmenting path of length one suffices to match $s$,
the time complexity is proportional to the degree of $s$, and hence summed over all 
unmatched vertices this is $O(m)$.  

Now we consider alternating paths of length $3$ (edges) from $s$; 
we search for an augmenting path among these.  Denote the number of such paths from $s$ by $L$. 
Let us denote a generic alternating  path of length three by $s$, $t_i$, $s_i$, $t_{L+i}$, 
for $i= 1$, $2$, $\ldots$, $L$. 
Furthermore, suppose one of these paths, $s$, $t_j$, $s_j$, $t_{L+j}$ is an augmenting path. 
After augmentation, we have the two matched edges $(s,t_j)$ and $(s_j, t_{L+j})$ 
and the unmatched edge $(t_j, s_j)$.

The cost of examining the neighbors of the unmatched vertices $s$ is clearly $O(m)$. 
Once we reach a matched neighbor $t_i$ of $s$, then we take the matched edge $(t_i, s_i)$,
and then search the neighbors of $s_i$ for an unmatched neighbor. 
Consider the neighbors of  the vertex $s_i$. 
If we find an unmatched neighbor $t_{L+i}$, then  we have an augmenting path, 
we match the edge $(s_i, t_{L+i})$, and we end the search.  
Once a neighbor of $s_i$ is matched, since it 
stays matched for the rest of the algorithm, we need not examine it again. 
If we find a matched neighbor of $s_i$,  then it cannot lead to an augmenting path of 
length three, and we can examine the next neighbor in the adjacency list. 
At each step, we can maintain a pointer to the first  unexamined neighbor of $s_i$ 
in the adjacency list of $s_i$ in the algorithm, and continue the search for an 
unmatched neighbor of $s_i$ from that vertex. 
This means that we go through the adjacency list of any  matched vertex in $S$  at most once,
and thus the cost of searching these vertices at a distance two edges from unmatched vertices in  $S$ 
is $O(m)$.

This completes the proof.
\end{proof}

\section{Correctness of the $2/3$-Approximation Algorithm}
\label{sec:sectionCorrectness-TwoThirdsAlg}

%

\begin{figure}[!thp]
\centering
\includegraphics[scale=0.08]{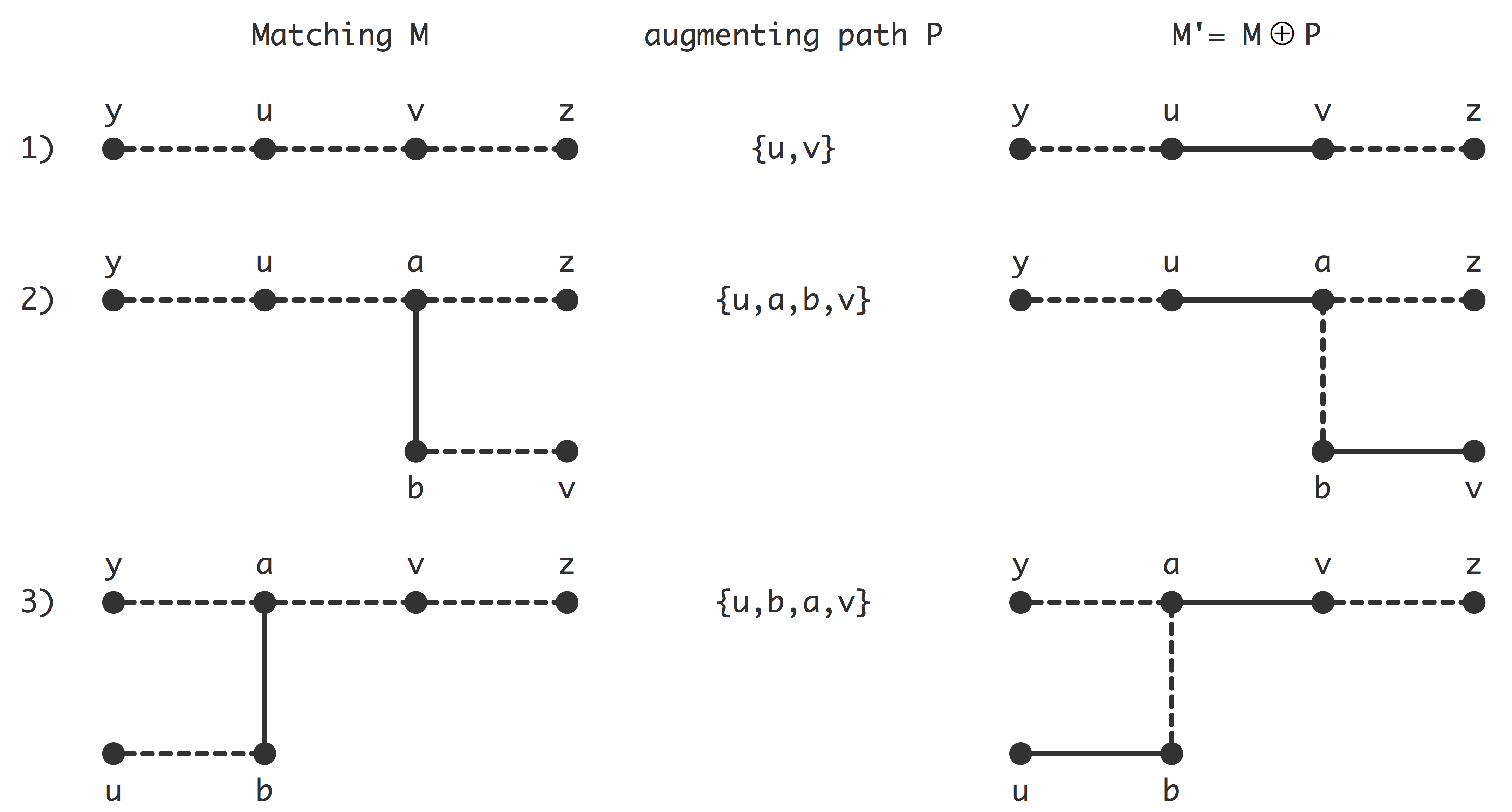}
\caption{ \textit{The three cases for augmenting paths of length one or three.} 
In each case, $M$ is a matching, $z$ is an $M$-unmatched vertex belonging to the vertex set $S$, 
$P$ is an $M$-augmenting path that joins $u$ (a heaviest $M$-unmatched vertex from $S$)
and $v$  (an arbitrary vertex in $T$ that can be reached by an $M$-augmenting path of length 
one or three from $u$). The augmented matching $M' = M \oplus P$. 
Matched edges are drawn as solid edges, and unmatched edges are drawn as dashed edges. 
In all three cases, the presence of an $M'$-augmenting path of length three from the vertex $z$ 
implies the presence of an $M$-augmenting path of length one or three from $z$ as well. 
}
\label{Fig:augpaths-lengththree}
\end{figure}

\begin{figure}[!thp]
\centering
\includegraphics[scale=0.08]{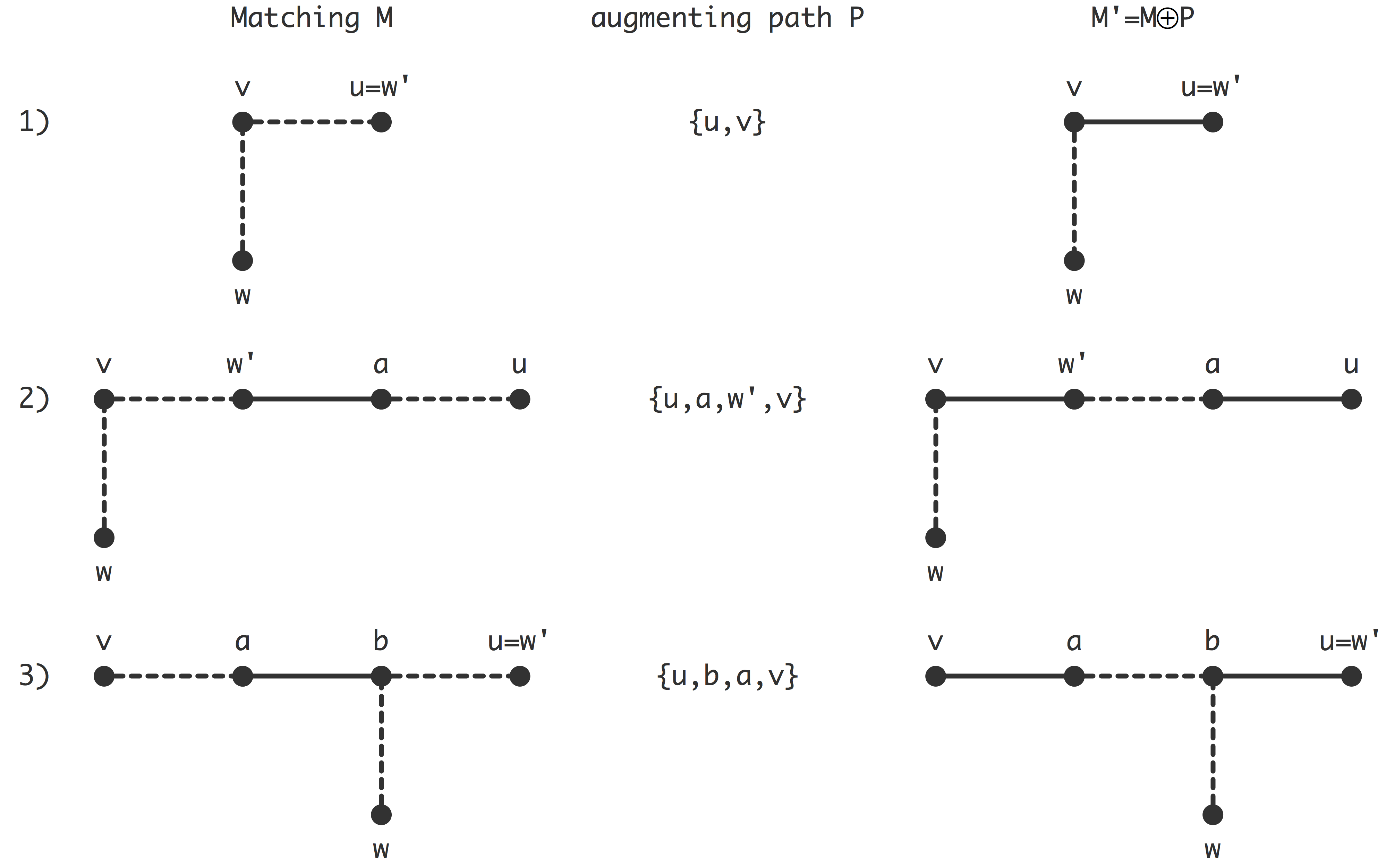}
\caption{\textit{The three cases for increasing paths of length two.} 
In each case, $M$ is a matching, and 
$P$ is an $M$-augmenting path that joins $u$ (a heaviest $M$-unmatched vertex from $S$)
and $v$  (an arbitrary vertex in $T$ that can be reached by an $M$-augmenting path of length 
one or three from $u$). The augmented matching $M' = M \oplus P$. 
Matched edges are drawn as solid edges, and unmatched edges are drawn as dashed edges. 
In all three cases, the absence  of an $M'$-increasing path of length two joining 
the vertices $w$ and $w'$ implies the absence  of an $M$-increasing path of length two as well. 
}
\label{Fig:incrpaths-lengthtwo}
\end{figure}

The following is the result we wish to prove in this Section. 
\begin{Theorem} 
\label{thm:two-thirdsapprox}
Let $G = (S, T, E)$ be a bipartite graph and $\phi: S \cup T \mapsto R_{\geq 0}$ a weight function.
Then Algorithm~\ref{alg:approx} computes a $2/3$-approximation for the MVM problem. 
\end{Theorem}
This is technically the most demanding section in this paper, and 
the reader could skip it in a first reading of the paper without loss of understanding. 
In order to prove this result, we need two supplementary results. 

\begin{Lemma}
\label{lem:short-aug-incr-paths}
Let $G = (S, T, E)$ be a bipartite graph,  
and $\phi: S \cup T \mapsto R_{\geq 0}$ be a weight function such
that $\phi(t) = 0$ for every vertex $t \in T$. 
Let $M$ be a matching in $G$, and $z \in S$ an $M$-unmatched vertex.
Suppose that (i) there is no $M$-augmenting path of length one or three from the vertex $z$, 
and  that (ii) there is no $M$-increasing path of length two in $G$.
Let $P$ denote an $M$-augmenting path of length at most  three 
with one end point $ u \in S$, a heaviest $M$-unmatched vertex,  and 
let the other endpoint of $P$ be $v \in T$. 
If  $M' = M \oplus P$ denotes  the augmented matching,  
then (i) there is no $M'$-augmenting path of length one or three from the vertex $z$,
and (ii) there is no $M'$-increasing path of length two  in $G$.  
\end{Lemma}

\begin{proof}

We will first consider the case of augmenting paths from the vertex $z$ 
of the specified length, and then consider increasing paths of length two. 

Since there is no $M$-augmenting path of length one from the $M$-unmatched vertex $z$, 
all neighbors of $z$ are matched under $M$. 
Since $P$ is an augmenting path, all vertices matched in $M$ continue to be matched in 
the matching $M' = M \oplus P$, and thus there cannot be an $M'$-unmatched edge incident on $z$,
and thus no $M'$-augmenting path of length one from $z$. 

Now suppose that there exists an $M'$-augmenting path of length three from the vertex $z$. 
Since no such path exists with respect to the matching $M$, the augmenting path $P$ 
must have some vertex adjacent to the vertex $z$. 
There are three possible cases  for the $M$-augmenting path $P$ of length one or three 
that joins a heaviest $M$-unmatched vertex $u$ to some vertex $v \in T$. 
The three cases are illustrated in Figure~\ref{Fig:augpaths-lengththree}. 

In the first case, $P= \{u, v\}$, in the second case $P= \{u, a, b, v\}$, and 
in the third case, $P=\{u, b, a, v\}$. 
In all three cases we can see that the existence of an $M'$-augmenting path of length one or three
from $z$ implies the existence of an $M$-augmenting path of length one or three from $z$ as well.
This contradiction proves the result regarding short augmenting paths. 

We turn to increasing paths of length two. 
Again, we suppose that there is an $M'$-increasing path of length two joining two vertices 
belonging to $S$ denoted by  $w$ and $w'$.
There are three cases to consider as illustrated in Figure~\ref{Fig:incrpaths-lengthtwo}. 

In the first case, since $u$ and $w$ are both $M$-unmatched, by choice of $u$ as a heaviest 
unmatched vertex, we have $\phi(w) \leq \phi(u) = \phi(w')$.
Hence the path $\{w, v, u=w'\}$ cannot be $M'$-increasing. 

In the second case, we have $\phi(w') \geq \phi(u)$ since there is no $M$-increasing path 
of length two. Since both vertices $u$ and $w$ are $M$-unmatched, we have $\phi(u) \geq \phi(w)$.
Combining the two inequalities, we have $\phi(w') \geq \phi(u) \geq \phi(w)$, and again
the path $\{w, v, w'\}$ is not an $M'$-increasing path. 

In the third case, we have $\phi(w) \leq \phi(u) = \phi(w')$ since both vertices  
$w$ and $u$ are $M$-unmatched. 
Then again, the path $\{w, b, u=w'\}$ cannot be $M'$-increasing. 

The contradictions obtained in all three cases complete the proof for short increasing paths. 
\end{proof}

To compare the weight of a maximum vertex-weighted matching $\overline{M}$ with another matching
$M$, we consider the symmetric difference of these two matchings. 
The subgraph induced by these two matchings 
consists of cycles and paths. 
Each cycle in this subgraph has all of its vertices matched in both matchings, 
so these do not contribute to the difference in their weights. 
Consider a path in this subgraph that begins with a vertex $u$ that is matched in 
the optimal matching $\overline{M}$ but not the suboptimal matching $M$. 
Here we choose $u$ to be a vertex that is weighted in the restricted bipartite matching problem.
If the path  has odd length, 
then it ends in a vertex $w$ also matched in $\overline{M}$ but not in $M$.
The vertex $w$ belongs to the unweighted vertex part. 
We call the  vertex $u$ a  {\em failure\/}, for the suboptimal algorithm failed to match it,
while the optimal algorithm succeeds in matching it,
and since  $u$ is responsible for the lower  weight of the the suboptimal matching $M$.   

In the subgraph considered above, 
we cannot have a path of odd length with both of its terminal vertices belonging to $M$ but  
not $\overline{M}$, for we could use such a  path to augment  the optimal matching $\overline{M}$. 
If a path beginning  with the vertex $u$ matched in $\overline{M}$ but not $M$ 
has even length, then it ends in a vertex $w$ matched in $M$ but not $\overline{M}$.  
If $\phi(u) >  \phi(w)$, then this path contributes to
a lower weight for the suboptimal matching $M$.  The approximation algorithm we 
have described does not permit the existence of increasing paths, 
and so we do not need to consider this here. 
We also cannot have $\phi(u) < \phi(w)$, for then we would have an $\overline{M}$-increasing path, 
contradicting the optimality of $\overline{M}$. 

We now focus on the vertices we have called failures. 
The idea is to show that failures are light and rare relative to 
other vertices matched in the suboptimal matching,
so that we can compensate for the failures through these vertices. 
For every failure, if we have a sufficiently large set of  compensating vertices $C(u)$,
and these sets of vertices are disjoint, then we can establish an approximation ratio 
for the suboptimal matching. 

\begin{Lemma}
\label{lem:general-approximation}
Let $G= (V,E)$ be a graph, $\phi: V \rightarrow R_{\geq 0}$ be a weight function, 
$\overline{M}$ an MVM in $G$,  $M$ any other matching,
and $h$ a positive integer. 
If for every failure $u$, there is a vertex-disjoint set of $h+1$ $M$-matched vertices $C(u)$ 
such that 
$\phi(u) \leq \phi(u')$ for all $u' \in C(u)$, 
then the matching $M$ is an $(h+1)/(h+2)$-approximate solution for the MVM problem. 
\end{Lemma}

\begin{proof}
Enumerate the failures as $u^j$, $=1$, $\ldots$, $n_f$,
and the set of compensating vertices for $u^j$ as $u^j_g$, for $g= 1$, $\ldots$, $h+1$. 
We can assume that all the compensating vertices are matched in $M$ and $\overline{M}$, 
which corresponds to the worst-case scenario for the approximation ratio. 

We consider the inequalities  that state that failures are light relative to their compensating vertices,  $\phi(u^j) \leq \phi(u^j_g)$  
for $g=1$, $\ldots$, $h+1$, and sum them over $g$,  to  obtain
\begin{displaymath}
(h+1) \phi(u^j)  \leq \sum_{g=1}^{(h+1)} \phi(u^j_g).
\end{displaymath}
We add $(h+1) \sum_{g=1}^{(h+1)} \phi(u^j_g)$ to both sides to obtain 
\begin{displaymath}
(h+1) \left[ \phi(u^j) + \sum_{g=1}^{h+1} \phi(u^j_g) \right]   
\leq (h+2) \sum_{g=1}^{(h+1)} \phi(u^j_g).
\end{displaymath}
Note that the left-hand-side of the inequality counts the weight of some of the matched vertices
in the optimal matching $\overline{M}$, 
and the right-hand-side counts the weight of some of the matched vertices in 
the suboptimal matching $M$. 
We sum this last inequality over all failures:
\begin{displaymath}
(h+1) \sum_{j=1}^{n_f} \left[ \phi(u^j) + \sum_{g=1}^{h+1} \phi(u^j_g)  \right] 
\leq (h+2) \sum_{j=1}^{n_f} \sum_{g=1}^{(h+1)} \phi(u^j_g).
\end{displaymath}
The vertices not included on either side of this inequality are vertices that are matched
in both matchings. We add $(h+1)$ times the sum of the weights of these latter vertices
to the left-hand-side and $(h+2)$ times this sum to the right-hand-side of the inequality, 
and obtain 
\begin{displaymath}
(h+1) \phi(\overline{M}) \leq (h+2) \phi(M).
\end{displaymath}
Rearranging, we find 
\begin{displaymath}
\phi(M) / \phi(\overline{M}) \geq (h+1)/(h+2).
\end{displaymath}

\end{proof}

We need to make an argument to charge the weight of a failed vertex 
to the set of compensating vertices, since these vertices are 
found from an alternating path constructed from the optimal matching and 
the current matching in the approximation algorithm. 
As the latter matching changes, the vertices on the alternating path from a  
failed vertex can change as well. 
The vertices  on the alternating path for the failure at a current step in the approximation 
algorithm might already have been charged for earlier failures, and  
hence we need a careful counting argument to find the set of compensating vertices
to charge for a failure.

{\bf Proof of Theorem~\ref{thm:two-thirdsapprox} }
\begin{proof}
Recall that we solve the problem by solving two separate matching problems, one 
with  weights only on the vertices in $S$ and the second, with weights only on the vertices in 
$T$. Using the Mendelsohn-Dulmage Theorem then 
we   combine the two matchings to find a matching that matches
all the matched vertices in $S$ in the first matching, and all the matched vertices in $T$
from the second matching. 

Let us consider the matching problems with weights on the vertex set $S$.
Let $M_S$ be the matching computed by the Approximation algorithm,
and $\overline{M}_S$ be a matching of maximum vertex weight.
We consider failures, i.e., vertices in $S$ that are matched in the optimal matching
but not in the approximate matching. 
We will show that every failure is compensated by \emph{two} vertices in $S$ that are matched 
by $M_S$ and are also heavier than the failed vertex. 
These sets of compensating vertices are vertex-disjoint, and this leads to the 
Two-third approximation. 

The Approximation algorithm considers vertices to match by non-increasing order of 
weights. If a short augmenting path (of length one or three) is found from an  unmatched vertex $u$,
then the algorithm augments the matching, and $u$ is matched.
If the Algorithm fails to find a short  augmenting path from $u$,
then it does not search for an augmenting path from the vertex $u$ again.
At the end of this step, we will say that the vertex $u$ has been {\em processed}. 

Let $n_a$ be the number of short augmenting operations in the Approximation algorithm, 
and let the matchings in the sequence of short augmentations be indexed as 
$M_i$, for $i=0$, $\ldots$, $n_a$. 
For $0 \leq i < n_a$, let $P_i$ denote the augmenting path used to augment the matching 
$M_i$ to $M_{i+1}$, and let $u_i$ denote the source of the augmenting path and 
$v_i$ denote its destination. 

First, we induct on the augmentation step $i$ to show that: 
\newline (1) no $M_i$-augmenting path of length one or three exists from any 
vertices that  are $M_i$-unmatched and have been processed prior to this augmentation step.
\newline (2) no $M_i$-increasing path of length two exists in $G$.

The base case is $i= 0$, and these results hold trivially since no vertex has been 
processed yet, and no vertices are matched. 
If the induction hypothesis holds at the beginning of augmentation step $i$, then 
by Lemma~\ref{lem:short-aug-incr-paths}
the hypothesis holds at the beginning of step $i+1$ as well, 
since we match a  currently heaviest unmatched and unprocessed vertex at this step.

When a vertex $u$ is marked as a failure (i.e., as a vertex that has been processed
and is $M$-unmatched), then the length of any augmenting or increasing path is at least four.
We will make use of this fact in a second inductive argument. 

We enumerate the failures in order of their processing time:
the vertex $u^k \in S$ is the $k$-th failure, and $n_f$ denotes the total number of 
failures. 
The second inductive argument is on the number of failures $k$, 
where $1 \leq k \leq n_f$.
Denote the matching at this step by $M_{f,k}$ (the matching associated
with the $k$-th failure). 
At step $k$, we consider {\em all failures\/}  up to this point, including $k$. 
\newline {\bf Claim}:  For every failure $u^j$ with $ 1 \leq j \leq k$, 
there are two $M_{f,k}$-matched vertices in $S$ labeled $u_1^j$ and $u_2^j$ that 
are heavier than $u^j$. Hence 
$\phi(u^j) \leq \phi(u_1^j)$, and $\phi(u^j) \leq \phi(u_2^j)$.

We prove the Claim by induction again. 
The base case of the induction hypothesis is $k=1$. 
Consider the situation  when the vertex $u^1$ is processed and is marked as a failure. 
The current matching is $M_{f,1}$, and we consider the symmetric difference 
$\overline{M}_S \oplus M_{f,1}$. 
The vertex $u^1$ is an endpoint for an alternating  path $P^{1,1}$ in the subgraph 
induced by the edges in the symmetric difference
(the edges belong  alternately to the matching $\overline{M}_S$ and $M_{f,1}$), 
and its length is at least four. 
Denote the vertex at distance two from $u^1$ by $u^1_1$ and 
the vertex at distance four from $u^1$ by $u_2^1$. 
These vertices are matched in $M_{f,1}$ and hence 
were processed earlier than $u^1$, and are hence at least as heavy as $u^1$. 
Thus the induction hypothesis holds for $k=1$.

Assume that the induction hypothesis is true for some $1 \leq k < n_f$, and consider 
the case for $k+1$, when a vertex $u^{k+1}$ is processed  and becomes a failure. 
The current matching is $M_{f, k+1}$, and by forming the symmetric difference
$\overline{M}_S \oplus M_{f, k+1}$ we see that every failure $u^j$ with $j=1$, $\ldots$, 
$k+1$ is an endpoint for an alternating  path $P^{j,k+1}$, whose length is at least four
(by the first inductive argument). 
Denote the vertices at distances two and four from $u^{j}$ by 
$u_1^{j,k+1}$ and $u_2^{j,k+1}$, respectively. 


When the graph $G$ is bipartite (which is the case here), 
we claim that these vertices form distinct pairs.
Consider an alternating  path from a failed vertex $u^j$ when edges are chosen 
from $\overline{M}_S \oplus M_{f, k+1}$. 
The vertex $u^j$ is matched in the optimal matching but not in the current matching 
$M_{f,k+1}$, and it belongs to $S$. 
Every vertex from $T$ reached by this alternating path is reached by an edge that belongs
to the optimal matching, while every  vertex in $S$ reached by this path 
(other than $u^j$) is reached by an edge that belongs to the current  matching. 
Hence it is clear that this path cannot reach another failed vertex $u^l$,
since such a vertex is not matched in the current  matching. 
Thus these sets of alternating paths from the failed vertices are vertex-disjoint. 

Define $A = \cup_{j=1}^{k+1} \{u_1^{j,k+1}, u_2^{j,k+1}\}$, and 
$B = \cup _{j=1}^k \{ u_1^j, u_2^j\}$. 
(The first set consists of vertices we find  from each failure 
using alternating paths from the optimal matching and the current matching.
The second set consists of the vertices that have been charged for the failures
prior to the current step.)
Then $|A| = 2(k+1)$ and $|B| = 2k$ since these elements are distinct. 
The set $B$ is not necessarily contained in the set $A$, and hence $ |A \setminus B | \geq 2$. 
Thus we can \emph{choose} two distinct vertices from the set $A \setminus B$ 
(and matched in $M_{f,k+1}$)  to associate with 
the failure $u^{k+1}$. Denote them by $u_1^{k+1}$ and $u_2^{k+1}$. 
Since these vertices are processed earlier than $u^{k+1}$, they are at least as 
heavy as $u^{k+1}$. Thus the induction hypothesis holds for the  $k+ 1$-st failure also. 

Now we can apply Lemma~\ref{lem:general-approximation} with $h=1$ to obtain the $2/3$-approximation bound. 

This completes the proof. 
\end{proof}

\section{A  $1/2$-Approximation Algorithm  for MVM in Bipartite Graphs}
\label{sec:sectionHalfApprox}

In this section we discuss a Greedy $1/2$-approximation algorithm 
for the vertex-weighted 
matching problem on bipartite graphs. 
Our intent is to compare the $2/3$-approximation algorithm from the previous Section
with this  algorithm, and hence our discussion will be brief. 
Also, the algorithm discussed here could be adapted to non-bipartite graphs
in a straightforward manner. The reason we discuss the  bipartite version here 
is that the specialized algorithm for bipartite graphs is  more efficient practically. 
The  algorithm solves  two Restricted Bipartite Matching 
(one-side weighted)  problems and then invokes  the Dulmage-Mendelsohn theorem as in the $2/3$-approximation algorithm. 

The Greedy $1/2$-approximation algorithm  has only one change
from Algorithm~\ref{Restr-Bipart-Approx}. 
In each iteration of the {\bf while} loop, 
it finds an unmatched neighbor of the currently heaviest unmatched vertex $u$
(an augmenting path of length $1$ instead of $3$).
Recall that since only one vertex set in the bipartite graph is weighted 
in the Restricted Bipartite Matching problem,
it can choose {\em any\/}  unmatched neighbor of  $u$, 
and does not need to look for the heaviest such neighbor. 

The following Lemma and Theorem show that this Greedy algorithm is a 
$1/2$-approximation algorithm for the VWM problem on bipartite graphs. 

\begin{lemma}
\label{lem:shorter-aug-incr-paths}
Let $G= (S, T, E)$ be  a bipartite graph, 
and $\phi: S \cup T \rightarrow R_{\geq 0}$ be a weight function such that 
$\phi(t) = 0$ for every vertex $ t \in T$. 
Let $M$ be a matching in $G$ and $z \in S$ be an $M$-unmatched vertex. 
Suppose that (i) there is no $M$-augmenting path of length one from $z$, and 
(ii) there is no $M$-increasing path of length two in $G$. 
Let $(u,v)$ denote an unmatched edge, where  $ u \in S$ is  
a heaviest $M$-unmatched vertex  and   $v \in T$. 
If  $M' = M \oplus \{(u,v)\}$ denotes  the augmented matching,  
then (i) there is no $M'$-augmenting path of length one from  the vertex $z$,
and (ii) there is no $M'$-increasing path of length two  in $G$.  
\end{lemma}

The proof of this Lemma is similar to Lemma~\ref{lem:short-aug-incr-paths},
and hence is omitted. 

\begin{Theorem} 
\label{thm:half-approx}
Let $G = (S, T, E)$ be a bipartite graph and $\phi: S \cup T \mapsto R_{\geq 0}$ a weight function.
Then the Greedy algorithm computes a $1/2$-approximation for the MVM problem on $G$. 
\end{Theorem}

The proof is by induction  on the number of augmentations,  using Lemma~\ref{lem:shorter-aug-incr-paths} at each augmenting step, and is similar to the proof of Theorem~\ref{thm:optimal}.  
It is again omitted.

\section{Experiments and Results}
\label{sec:sectionExperiments}

\subsection{Experimental Setup}
For the experiments, we used an
Intel Xeon E5-2660 processor based system (part of the  Purdue University Community Cluster), called \emph{Rice}. 
The machine consists of two processors, each with ten cores running at 2.6 GHz 
(20 cores in total) with $25$ MB unified L3 cache and 64 GB of memory. 
The operating system is Red Hat Enterprise Linux release 6.9. 
All  code was developed using C++ and compiled using
the {\tt g++} compiler (version: $4.4.7$) using the -O3 flag.

 \begin{table}
 {
 \caption{\textit{The set of bipartite graphs which are our test problems.} 
The problems are listed in increasing order of the total number of vertices. }
 \footnotesize
 \centering
 \begin{tabular}[!h]{|l|rrr|rrr|r|}
 \hline
 \textbf{Graph} &$|V_1|$ &\multicolumn{2}{c|}{Degree}&$ |V_2|$&\multicolumn{2}{c|}{Degree} &  $|E|$ \\
               &         & \textbf{Max.} & \textbf{Mean}&    & \textbf{Max.} & \textbf{Mean} 
                                                       & \\
 \hline \hline
Trec10  & 106	&134    &	81.25&	478&	79&	18.02&	8,612 \\
IG5-16  & 18,485 &	990 & 31.83	& 18,846 &	120 &	31.22 &	588,326 \\
fxm3\_16 & 41,340 &	57 &	9.49 &	85,575 &	36 &	4.58 &	392,252 \\
JP & 67,320 &	8,980 &	204.02 &	87,616 &	390 &	156.76 & 13,734,559 \\
flower\_8\_4 & 55,081 & 15 &	6.81 &	125,361 &	3 &	2.99 &	375,266 \\
spal\_004 &	10,203 & 6,029 &	4524.96 &	321,696 &	168 &	143.52 & 46,168,124 \\
pds-50 & 83,060 & 96 &	7.11 &	275,814 &	3 &	2.14 &	590,833 \\
image\_interp &  120,000 &	6 &	5.93 &	240,000 &	5 &	2.97 &	711,683 \\
kneser\_10\_4\_1 &	330,751 &	3 &	3.00 &	349,651 &	16 &	2.84 &	992,252 \\
12month1 &	12,471 &	75,355 &	1814.19 &	872,622 &	3,420 &	25.93 &	22,624,727 \\
IMDB &	428,440 &	1,334 &	8.83 &	896,308 &	1,590 & 	4.22 &	3,782,463 \\
GL7d16 & 460,261 &	114 &	31.48 &	955,128 &	64 & 15.17 & 14,488,881 \\
wheel\_601 &	723,605 &	3 &	3.00 &	902,103 &	602 &	2.41 &	2,170,814 \\
Rucci1 &  109,900 &	108 &	70.89 &	1,977,885 &	4 &	3.94 &	7,791,168 \\
LargeRegFile & 801,374 &	655,876 &	6.17 &	2,111,154 &	4 &	2.34 &	4,944,201 \\
GL7d20 &    1,437,547 &	395 &	20.79 &	1,911,130 &	43 &	15.64 &	29,893,084 \\
GL7d18 &    1,548,650 &	69 &	22.98 &	1,955,309 &	73 &	18.20 &	35,590,540 \\
GL7d19 &    1,911,130 &	121 &	19.53 &	1,955,309 &	54 &	19.09 &	37,322,725 \\
relat9 &	549,336 &	227 &	70.91 &	12,360,060 &	4 &	3.15 &	38,955,420 \\
 \hline
 \end{tabular}
 \label{Table:Problems}
 }
 \end{table}

Our test set  consists of nineteen real-world bipartite  graphs taken from the  
University of Florida Matrix collection \cite{FMC11}
covering several application areas. 
We chose the largest rectangular matrices in the collection, and then added a few smaller matrices. 
Table \ref{Table:Problems} gives some statistics  on  our test set.
The bipartite graphs are listed in increasing order of the total number of vertices.
The largest number of vertices of any graph is  nearly $13$ million, 
and the largest number of edges is $46$ million. 
For each vertex set $V_i$ in the bipartition we list the cardinality of the set,
and the maximum and average vertex degrees. 
The average degree  varies from two  to four  thousand,
and hence the graphs are diverse with respect to their degree distributions. 
The weights of the vertices were generated as random integers in the range $[1, 1000]$. 
We compare the performance of six different exact and approximate matching algorithms. Two of these are  an exact maximum edge-weighted matching 
algorithm (MEM), and an exact maximum vertex-weighted matching algorithm (MVM). 
Two are the $2/3$- and $1/2$-approximate MVM algorithms based on finding short augmenting paths that we have discussed in this paper. 
The last two algorithms are obtained from the $(1-\epsilon)$-approximation algorithm for MEM based on the scaling approach of Duan and Pettie,
where we have chosen $\epsilon$ equal to $1/3$ and $1/6$ to obtain $2/3$- and $5/6$-approximation algorithms. 

The MEM algorithm is a primal-dual algorithm for sparse bipartite graphs with 
$O(n m \log n)$ time complexity~\cite{galil1},
which has been implemented in the {\tt Matchbox}  software by our research group.  
We apply the MEM algorithms to the vertex-weighted matching problems by  assigning to each edge the sum of the weights of its endpoints. 
The Exact MVM algorithm we have implemented is Algorithm~\ref{Exact-General-VWM},
and not the Spencer and Mayr algorithm, for the following reasons: The former algorithm is easy to implement, and has good performance, while the latter is more complicated to implement.  As can be seen from the earlier work on matchings discussed in Section~\ref{section:introduction}, asymptotically fastest algorithms are not necessarily the fastest algorithms in  practice.  
Finally, our focus in this paper is on the approximation algorithms. 

\begin{table}
 {
  \caption{\textit{Comparing the  weight of the matchings
  computed by six different algorithms. The Exact MVM and MEM algorithms compute
  the same matching, and for these we report the absolute values of these 
  quantities. The results of the four approximation algorithms are reported as the ratio of the weight  to the weight of the 
  exact algorithms. 
   } }
 \footnotesize
 \centering

 \begin{tabular}[!h]{|l|r|r|r|r|r|}
 \hline
 \textbf{Graph}  & Exact Algorithms & \multicolumn{4}{c|}{Appr. Algorithms}\\
        &  & \multicolumn{2}{c|}{Aug. Path Approach}  & \multicolumn{2}{c|}{Scaling Approach} \\
 &  absolute	 &	2/3-Appr. & 1/2-Appr. &	 2/3-Appr. & 	5/6-Appr.\\
                & weight & \multicolumn{4}{|c|}{relative weight} \\                         
 \hline \hline
Trec10	&	1.43E+04	&	0.999	&	0.988	&	0.942	&	0.984	\\
IG5-16	&	1.16E+07	&	0.987	&	0.933	&	0.906	&	0.928	\\
fxm3\_16	&	5.06E+07	&	0.995	&	0.963	&	0.962	&	0.971	\\
JP	&	3.42E+07	&	0.989	&	0.956	&	0.924	&	0.956	\\
flower\_8\_4	&	6.99E+07	&	0.990	&	0.963	&	0.958	&	0.969	\\
spal\_004	&	1.51E+07	&	1.000	&	0.996	&	0.880	&	0.923	\\
pds-50	&	1.06E+08	&	0.996	&	0.980	&	0.953	&	0.968	\\
image\_interp	&	1.48E+08	&	0.993	&	0.965	&	0.933	&	0.946	\\
kneser\_10\_4\_1	&	3.36E+08	&	0.996	&	0.960	&	0.962	&	0.964	\\
12month1	&	1.82E+07	&	0.999	&	0.991	&	0.875	&	0.921	\\
IMDB	&	3.04E+08	&	0.987	&	0.927	&	0.927	&	0.940	\\
GL7d16	&	5.76E+08	&	0.995	&	0.942	&	0.988	&	0.994	\\
wheel\_601	&	7.84E+08	&	0.990	&	0.903	&	0.931	&	0.947	\\
Rucci1	&	1.62E+08	&	0.999	&	0.997	&	0.918	&	0.954	\\
LargeRegFile	&	9.72E+08	&	0.998	&	0.979	&	0.957	&	0.969	\\
GL7d20	&	1.61E+09	&	0.998	&	0.948	&	0.990	&	0.994	\\
GL7d18	&	1.70E+09	&	0.993	&	0.921	&	0.995	&	0.996	\\
GL7d19	&	1.92E+09	&	0.994	&	0.926	&	0.994	&	0.995	\\
relat9	&	4.08E+08	&	1.000	&	0.999	&	0.910	&	0.948	\\[1ex]
\hline
Geom. Mean	&	1.00	&	0.995	&	0.960	&	0.942	&	0.961	\\

 \hline
 \end{tabular}
 \label{table:Weight}
 }
 \end{table}
 
 \begin{table}
 {
  \caption{\textit{Comparing the cardinality  of the matchings
  computed by six different algorithms. The Exact MVM and MEM algorithms compute
  the same matching, and for these we report the absolute values of these 
  quantities. The results of the four approximation algorithms are reported as the ratio of the cardinality to cardinality of the 
  exact algorithms. 
   } }
 \footnotesize
 \centering

 \begin{tabular}[!h]{|l|r|r|r|r|r|}
 \hline
 \textbf{Graph}  & Exact Algorithms & \multicolumn{4}{c|}{Appr. Algorithms}\\
        &  & \multicolumn{2}{c|}{Aug. Path Approach}  & \multicolumn{2}{c|}{Scaling Approach} \\
 &  absolute &	2/3-Appr. & 1/2-Appr. &	 2/3-Appr. & 	5/6-Appr.\\
                & cardinality & \multicolumn{4}{|c|}{relative cardinality} \\                        
 \hline \hline
Trec10	&	 106 	&	1.000	&	1.000	&	1.000	&	1.000	\\
IG5-16	&	 9,519 	&	1.000	&	1.000	&	0.953	&	0.965	\\
fxm3\_16	&	 41,340 	&	1.000	&	0.999	&	0.994	&	0.996	\\
JP	&	 26,137 	&	0.994	&	0.979	&	0.978	&	0.980	\\
flower\_8\_4	&	 55,081 	&	1.000	&	0.999	&	0.997	&	0.998	\\
spal\_004	&	 10,203 	&	1.000	&	1.000	&	1.000	&	1.000	\\
pds-50	&	 82,837 	&	1.000	&	1.000	&	0.996	&	0.996	\\
image\_interp	&	 120,000 	&	1.000	&	1.000	&	0.999	&	0.999	\\
kneser\_10\_4\_1	&	 323,401 	&	0.999	&	0.969	&	0.970	&	0.974	\\
12month1	&	 12,418 	&	1.000	&	1.000	&	0.999	&	0.999	\\
IMDB	&	 250,516 	&	0.992	&	0.958	&	0.955	&	0.962	\\
GL7d16	&	 460,091 	&	1.000	&	0.999	&	1.000	&	1.000	\\
wheel\_601	&	 723,005 	&	1.000	&	0.930	&	0.940	&	0.960	\\
Rucci1	&	 109,900 	&	1.000	&	1.000	&	1.000	&	1.000	\\
LargeRegFile	&	 801,374 	&	1.000	&	0.999	&	0.998	&	0.999	\\
GL7d20	&	 1,437,546 	&	1.000	&	0.992	&	0.999	&	1.000	\\
GL7d18	&	 1,548,499 	&	1.000	&	0.974	&	0.999	&	0.999	\\
GL7d19	&	 1,911,130 	&	0.998	&	0.920	&	0.965	&	0.971	\\
relat9	&	 274,667 	&	1.000	&	1.000	&	1.000	&	1.000	\\[1ex]
\hline
Geom. Mean	&	1.00	&	0.999	&	0.985	&	0.986	&	0.989	\\
 \hline
 \end{tabular}
 \label{table:Cardinality}
 }
 \end{table}

We compare the weights computed by the six algorithms in 
Table~\ref{table:Weight}. Both exact MEM and MVM algorithms compute 
the same matching, and hence we report one set of weights and cardinalities for these algorithms. We report the absolute weight obtained by the exact algorithms,
and for the approximation algorithms report the 
fraction of the maximum weight obtained by them. 
We report the cardinality of the matchings computed by the 
six algorithms in Table~\ref{table:Cardinality}. 
The results are reported in a format similar to that for the weights. 

Ten of these graphs have their MVM corresponding to 
$V_1$-perfect matchings, i.e., the cardinality of the MVM is equal to the 
cardinality of the smaller vertex set $V_1$. 
There are also four graphs where the cardinality is lower than almost half  the value of $|V_1|$: IG5-16, JP, IMDB, and relat9. 
All of the four approximation algorithms compute weights that are higher  than $90\%$ of the maximum weight obtained by the exact algorithm (with two exceptions),
much higher than the guaranteed approximation ratios ($1/2$, $2/3$, or $5/6$). When we consider the geometric means, the $2/3$-approximation 
MVM algorithm obtains $99.5\%$ of the weight, 
and $99.9\%$ of the cardinality of the maximum weight matching. 
The $1/2$-approximation MVM algorithm obtains values that are  lower, 
$96\%$ for the weight and $98.5\%$ for the cardinality. 
The scaling algorithms perform worse than the augmenting path-based approximation  algorithms: even the $5/6$-approximation scaling algorithm
obtains only $96.1\%$ of the weight and $98.9\%$ of the cardinality,
values that are lower than the $2/3$-approximation MVM. 
The relative weights of the matchings computed by  three of the approximation algorithms are plotted in Figure~\ref{Fig:weights}. The problems are listed in
order of increasing relative weight of the $2/3$-approximation MVM algorithm. 

We compare the run-times of the Exact MEM, Exact MVM, and the four  approximation algorithms in Table~\ref{Table:RunningTime}.
The time (in seconds) taken by the Exact MEM algorithm to compute the maximum weight matching is reported; for the other five algorithms, we report the relative performance, which is the ratio of the run-time of the 
MEM algorithm to the run-time of each of the other algorithms. Thus the values in the Table are proportional to the reciprocal running time, and higher the value, the faster the algorithm. 

The Exact MEM  algorithm is fast for the smaller problems, but as the number of vertices and edges gets into the tens of millions, it can require more than $15$ hours (on graph GL7d18) to compute the matching. The maximum time needed by the Exact MVM algorithm on any graph is $22$ minutes for GL7d18 again;  the $2/3$-approximation MVM algorithm 
takes the maximum time of $4.7$ seconds on the GL7d19 graph; the scaling algorithms can take 1 minute for relat9 ($5/6$-approximation) and 40 seconds ($2/3$-approximation) for the same problem. 
In terms of geometric means, the exact MVM algorithm is $12$ times faster
than the exact MEM algorithm, while the $2/3$-approximation MVM is 
$63$ times faster, and the $1/2$-approximation algorithm is $100$ times faster, both relative to the Exact MEM algorithm. 
The scaling algorithms are only $7$ times faster ($2/3$-approximation)
and $5$ times faster ($5/6$-approximation) than the exact MEM algorithm. 
Note also that in general the run-times increase with the size of the graph, but they also depend on  how the edges are distributed within the graph. 

These results are plotted in Figure~\ref{Fig:performance}, where the $y$-axis is in logarithmic scale. Note that for the four largest graphs (the three GL7d graphs and relat9), the $2/3$-approximation MVM algorithm is more than $1,000$ times faster than the Exact MEM algorithm. The scaling algorithms are slower than it, but they perform relatively better for larger graphs relative to the Exact MEM algorithm, while they are slower for the smaller graphs. The Exact MVM algorithm tracks the $2/3$-approximation MVM algorithm for smaller graphs, but it is much slower  for the larger problems. 

\chg{The exact algorithms for both MEM and MVM are our implementations, and we spent a reasonable amount of effort to make them efficient; however, these are more sophisticated than the approximation algorithms, which are simpler to implement. Hence 
it might be possible to make the exact algorithms faster with optimizations that we have not considered, but the lower asymptotic time complexities of the approximation algorithms will make them faster in practice, as our results show.}


\begin{figure}[!thp]
\centering
\includegraphics[scale=0.65]{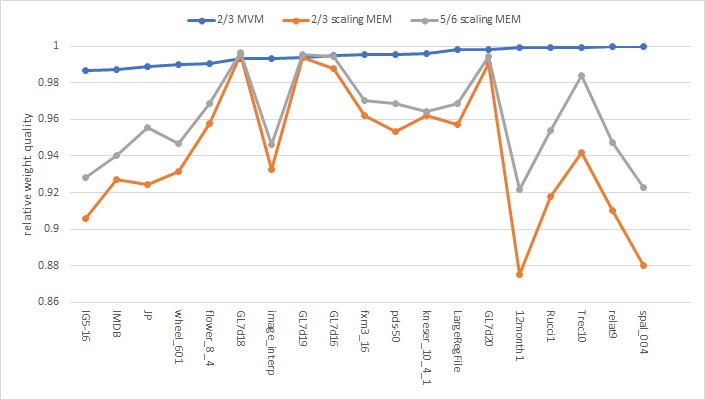}
\caption{ \textit{The weights of the matchings computed by the approximation
algorithms relative to the Exact MVM algorithm. 
The $x$-axis lists the nineteen problems in increasing order of the relative weight obtained by the $2/3$-approximation MVM algorithm, and the $y$-axis shows the percentage of the 
weight  obtained by three approximation algorithms. 
}}
\label{Fig:weights}
\end{figure}

 \begin{table}
 {
  \caption{\textit{The relative performance of the Exact MVM 
  algorithm and the four approximation algorithms relative to the 
  Exact MEM algorithm. We report run-times of the Exact MEM algorithm
  in seconds; for all others, we report the relative performance, which is the ratio of the runtime of the Exact MEM algorithm to that of the 
  other algorithms. Hence the value for an algorithm  shows how fast it is 
  relative to the Exact MEM algorithm.
 }}
 \footnotesize
 \centering

 \begin{tabular}[!h]{|l|r|r|r|r|r|r|}
 \hline
 \textbf{Graph}  & \multicolumn{2}{c|}{Exact Algorithms} & \multicolumn{4}{c|}{Approx. Algorithms}\\
        & MEM & MVM & \multicolumn{2}{c|}{Aug. Path Approach}  & \multicolumn{2}{c|}{Scaling Approach} \\
         &    &     &  2/3-Appr. & 	1/2-Appr. &	 2/3-Appr. & 	5/6-Appr.\\
& Time (s)& \multicolumn{5}{|c|}{Relative Performance} \\
   \hline \hline
Trec10	&	2.36E-3	&	16.71	&	24.17	&	32.11	&	1.67	&	0.87	\\
IG5-16	&	1.50	&	17.84	&	171.82	&	274.39	&	14.09	&	9.54	\\
fxm3\_16	&	6.95E-2	&	1.58	&	2.26	&	3.18	&	0.42	&	0.29	\\
JP	&	46.82	&	15.75	&	322.49	&	604.48	&	16.50	&	10.62	\\
flower\_8\_4	&	0.83	&	8.16	&	18.91	&	25.92	&	2.67	&	1.86	\\
spal\_004	&	53.57	&	71.87	&	150.48	&	218.84	&	6.36	&	4.54	\\
pds-50	&	0.12	&	0.84	&	1.32	&	1.75	&	0.14	&	0.10	\\
image\_interp	&	0.13	&	0.86	&	1.25	&	1.61	&	0.23	&	0.16	\\
kneser\_10\_4\_1	&	1.01	&	2.51	&	4.36	&	5.09	&	0.77	&	0.54	\\
12month1	&	28.24	&	35.90	&	78.27	&	117.15	&	4.70	&	3.12	\\
IMDB	&	3.09	&	3.74	&	6.88	&	10.93	&	0.57	&	0.37	\\
GL7d16	&	6.63E+3	&	98.86	&	7.39E+3	&	1.39E+4	&	725.45	&	397.21	\\
wheel\_601	&	0.72	&	0.62	&	1.17	&	1.51	&	0.18	&	0.12	\\
Rucci1	&	41.03	&	21.42	&	65.23	&	93.67	&	6.93	&	4.82	\\
LargeRegFile	&	2.08	&	1.17	&	1.97	&	2.79	&	0.53	&	0.38	\\
GL7d20	&	1.85E+4	&	385.41	&	5.78E+3	&	1.37E+4	&	746.76	&	482.75	\\
GL7d18	&	5.02E+4	&	37.97	&	1.34E+4	&	3.43E+4	&	1.85E+3	&	1.20E+3	\\
GL7d19	&	2.57E+4	&	162.21	&	5.47E+3	&	1.53E+4	&	1.00E+3	&	689.35	\\
relat9	&	3.63E+3	&	85.38	&	877.71	&	1.35E+3	&	87.64	&	60.44	\\[1ex]
\hline
Geom.~Mean&	1.00 &	12.02	&	62.64	&	99.66	&	6.99	&	4.64	\\
 \hline
 \end{tabular}
 \label{Table:RunningTime}
 }
 \end{table}

\begin{figure}[!thp]
\centering
\includegraphics[scale=0.65]{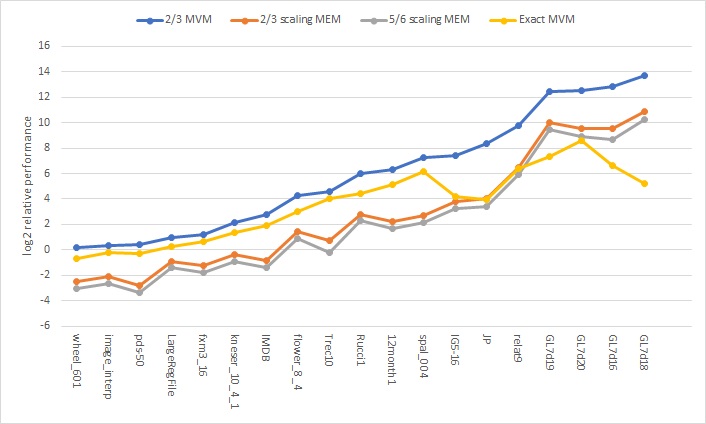}
\caption{ \textit{The relative performance (proportional to the reciprocal run-time) of the Exact MVM algorithm and the four approximation algorithms relative to the Exact MEM algorithm. 
The $x$-axis lists the nineteen problems in increasing order of the relative performance of the $2/3$-approximation MVM algorithm, and the $y$-axis shows the relative performance of the algorithms 
in logarithmic scale. 
}}
\label{Fig:performance}
\end{figure}

\begin{table}
\begin{center} 
 {
 \chg{
  \caption{\textit{The relative performance of the Exact 
  MEM algorithm for four choices of weights for the {\tt GL7d20} graph. 
  The first set of weights correspond to the original matrix values, the second and third to random edge weights in the ranges shown, and the fourth 
  to random vertex weights in $[0.5, 1000]$ that are summed to create edge
  weights. The results show that the fourth choice of weights
  leads to large runtimes due to the greatly increased length of the augmenting path searches. 
 }}
 \footnotesize
 \begin{tabular}[!h]{|l|r|r|r|r|r|}
 \hline
 & \multicolumn{5}{|c|}{Random Weights} \\
 \textbf{Metric}  & Original & edge & edge &  vertex & vertex \\ 
            &   & $[1, 2000]$ & $[0, 1]$ & $[0.5, 1000]$  & $[0, 0.5]$ \\
    &  &  &  &summed & summed \\
   \hline 
Time (s)	& 4.61 E0	& 1.366 E1	&	1.330 E1 & 1.831 E4 & 1.001 E5
\\ 
Cardinality & 1,437,546	& 1, 437, 545 & 1,437,546	&1,437,546,  &1,437,546\\ 	
No. augmentations  & 1,437,546	& 1, 437, 545 & 1,437,546	&1,437,546  & 1,437,546\\ 	
Aug.~path lengths & & & & & \\
Maximum &  9  & 55 & 61 & 181 & 2383\\
No.~distinct  & 5 & 28 & 30 & 74 & 938 \\
Mean  & 1.009 & 1.896 & 1.896 & 7.126 & 72.55\\[1ex] 
No.~dual updates & 1.350 E7 & 7.499 E6 & 7.485 E6 & 2.728 E10 & 3.921 E10\\
Time aug.~paths  (s) & 3.38 & 9.67 & 9.72 & 1.128 E4 & 4.314 E4
\\
Time dual updates (s) & 0.42 & 0.26 & 0.27 & 2.369 E3 & 4.435 E3\\
 \hline
 \end{tabular}
 \label{Table:GL7d20-weights}
 }
 }
 \end{center} 
 \end{table}

\chg{The large running times of the exact MEM algorithms on MVM  problems is due to the mismatch between algorithm and problem. Adding vertex weights to create edge weights causes edges incident on a vertex  to have highly correlated weights,  and this increases the average length of the augmenting paths  searched in the course of the algorithm. Table~\ref{Table:GL7d20-weights} 
shows various metrics for the {\tt GL7d20} matrix, one of the matrices for which the MEM algorithm takes a long time.
We show what happens to several metrics when we use the original matrix weights, two sets of random values for the edge weights, and two sets of random values for the vertex weights that are  summed
to create edge weights. We have used two ranges of weights, 
integers in the range $[1\ 2000]$, and real numbers in $(0 \ 1]$ to see how the range of weights influences the runtimes of the algorithms. 
Note that the runtime needed for the first three choices of weights is three to four orders of magnitude smaller than the times for the 
last two sets of weights. We also break down the time taken for the augmenting path searches and the dual weight updates. 
The runtime is largely accounted for by the time taken for augmenting path searches for every experiment but the last. The large increase in the runtime for random vertex weights is caused by the increase in the average length of an augmenting paths, and  for the case of real-valued weights, the larger time to process real-valued dual weights. 
Note  that the range of weights does not influence the runtimes significantly for edge weights. 
However, when real-valued weights in $(0\ 0.5])$ are used for vertex weights, the algorithm requires nearly $28$ hours! In this case, integer weights cause the algorithm to take $5$ hours. 
We leave a more thorough evaluation of all the 
causes for this for the future. 
However, these results support our contention that MVM problems should be solved by algorithms designed specifically for them,  rather than by converting them to MEM problems and then using MEM algorithms. 
}

\section{Conclusions}
\label{sec:sectionConclusions}

We have described a $2/3$-approximation algorithm for 
MVM  in bipartite graphs whose time complexity is $O(n \log n + m)$,
whereas the time complexity of an Exact algorithm is $O(n^{1/2} m \log n)$. 
The algorithm exploits the bipartiteness  of the graph
by decomposing the problem into two `one-side-weighted' problems, solving them 
individually,  and then combining the two matchings into the final matching. 
The algorithm also  sorts  the weights, processing the unmatched vertices in non-decreasing order of weights. 
The approximation algorithm is derived in a natural manner from an  algorithm 
for computing maximum weighted matchings by restricting the length of 
augmenting paths to at most three.

The $2/3$-approximation algorithm has been implemented in C++, and on a set of nineteen graphs, some with millions of vertices and edges, it computes the approximate matchings in less than $5$ seconds on a desktop processor. The weight of the approximate matching is greater than  $99\%$ of the weight of the optimal matching for these problems. A Greedy $1/2$-approximation algorithm is faster than the $2/3$-approximation algorithm by about a   factor of $1.5$, but the weight it computes is lower, and can be as low as $90\%$ on the worst problem. 
A path on four vertices $P_4 = \{v_1, v_2, v_3, v_4\}$,  where the sum of the weights of  vertices  $v_2$ and  $v_3$ is the heaviest over all consecutive pairs of vertices, is a contrived worst-case example for the Greedy algorithm,  but the $2/3$-approximation algorithm computes the maximum weight matching. 
Several copies of the  path $P_4$ can be joined together in a suitable manner to construct larger graphs where the same property holds. 
Whether the the Greedy $1/2$- or the $2/3$-approximation algorithm is to be preferred, trading off run time for increased weight, would depend on the context in which it is being used.  For example, it is known from experience that the $1/2$-approximation MEM algorithms do not lead to good orderings for sparse Gaussian elimination.  Recent work suggests that implementations of the $2/3-\epsilon$-approximation algorithm leads to better matrix orderings for this problem~\cite{Azad+:Arxiv}. 

The Exact MVM algorithm  shows that  the ``structure'' of the vertex weighted matching problem is  closer to the maximum cardinality matching problem rather than the maximum edge-weighted matching problem (MEM), in that we do not need to 
invoke linear programming duality,  and compute and update dual weights. 

We have also implemented the 
$(1- \epsilon)$-approximation algorithm for maximum edge-weighted matching, based on scaling the weights,  designed by Duan and Pettie~\cite{Duan+:jacm}. 
This algorithm is quite sophisticated, and can be applied to the MVM problem by transforming it into an MEM problem. 
However, our results show that it is an order of magnitude or more slower than the $2/3$-approximation algorithm for MVM; 
it also obtains lower weights for the approximate matching, even when we seek a $5/6$-approximation. 
An approximation algorithm for MEM 
analogous to the  $2/3$-approximation algorithm for MVM is not known that works on augmenting path lengths. 

Recent developments in  half-approximation algorithms for MEM (e.g., the Locally Dominant edge and Suitor algorithms~\cite{MH14}) show that we should be able to use
these algorithms that avoid sorting and obtain $1/2$-approximations for the MVM. Could similar algorithms be 
developed for the MVM problem to obtain 
$2/3$-approximation  without sorting? 
The speculative approach to solving the MVM problem employs a different strategy by first computing a MCM and then using increasing paths to improve the weight. This is the scope of our current work. 

The proof technique used  in this paper cannot be extended to obtain 
approximation ratios higher than $2/3$, since Lemma~\ref{lem:short-aug-incr-paths}
does not hold for higher (augmenting or increasing) path lengths.
Consideration of the MEM problem again suggests that 
there are other approaches that would lead to better approximation ratios, although they might not necessarily lead to practical improvements, given the high matching weights obtained from the approximation algorithms discussed here. 

Finally, we believe that the idea  of restricting the augmenting path length 
to at most three could lead to a $2/3$-approximation algorithm for non-bipartite graphs, 
although we will no longer be able to invoke the Mendelsohn-Dulmage theorem, and a different proof technique will be required. 



\bibliographystyle{siamplain}
\bibliography{refs1,refs2}
\end{document}